\documentclass[11pt]{article}
\usepackage{graphicx}
\usepackage{epstopdf}
\usepackage{amssymb,bm}
\usepackage{bbm}
\usepackage{amsmath}
\usepackage{amsfonts}
\usepackage[margin=1.0in]{geometry}
\usepackage{cite}
\usepackage{enumitem}

\usepackage{color}

\DeclareMathOperator{\E}{\mathbbmss{E}}
\DeclareMathOperator{\Prob}{\mathbbmss{P}}
\DeclareMathOperator{\Trn}{\mathsf{T}}

\allowdisplaybreaks

\usepackage{amsthm}
\newtheorem{theorem}{Theorem}
\newtheorem{definition}{Definition}
\newtheorem{lemma}{Lemma}

\newtheorem{remark}{Remark}

\usepackage{chngcntr}
\counterwithout{theorem}{section}
\counterwithout{definition}{section}
\counterwithout{lemma}{section}
\counterwithout{remark}{section}
\counterwithout{assumption}{section}
\counterwithout{proposition}{section}
\counterwithout{corollary}{section}
\counterwithout{claim}{section}

\begin{document}
\title{On the Separability of Parallel MISO Broadcast Channels Under Partial CSIT: A Degrees of Freedom Region Perspective
\footnotetext{H. Joudeh is with the  Communications and Information Theory Group, Faculty of Electrical Engineering and Computer Science, Technische Universit\"{a}t Berlin, 10587 Berlin, Germany (e-mail: h.joudeh@tu-berlin.de).
B. Clerckx is with the Communications and Signal Processing Group, 
Department of Electrical and Electronic Engineering, Imperial College London, London SW7 2AZ, U.K. 
(e-mail: b.clerckx@imperial.ac.uk).
\\
This work was partially supported by the U.K. Engineering and Physical Sciences Research Council (EPSRC)
under grants EP/N015312/1 and EP/R511547/1. Parts of this paper were presented at the 2017 IEEE GLOBECOM \cite{Joudeh2017} and the 2019 IEEE SPAWC \cite{Joudeh2019}.}}
\author{Hamdi~Joudeh and Bruno~Clerckx}
\date{}
\maketitle
\begin{abstract}
We study the $K$-user, $M$-subchannel parallel multiple-input-single-output (MISO) broadcast channel (BC)
under arbitrary levels of partial channel state information at the transmitter (CSIT).
We show that the parallel subchannels constituting this setting are separable from a degrees-of-freedom (DoF) region perspective
\emph{if and only if} the partial CSIT pattern is \emph{totally ordered}. This total order condition corresponds to users abiding by the same order, with respect to their CSIT
quality levels, in each of the parallel subchannels.
For instance, let $\alpha_{k}^{[l]}$ and $\alpha_{j}^{[l]}$ be the CSIT quality parameters for users $k$ and $j$ over subchannel $l$. Under total order,  
having $\alpha_{k}^{[l]} \geq \alpha_{j}^{[l]}$ implies that $\alpha_{k}^{[m]} \geq \alpha_{j}^{[m]}$ holds for every subchannel $m$.
In this case, the entire DoF region is achievable using simple separate coding, where a single-subchannel-type transmission scheme is employed in each subchannel.
To show this separability result, we first derive an outer bound for the DoF region by extending the aligned image sets approach of Davoodi and Jafar
to the considered setting.
We then show that this outer bound coincides with the inner bound achieved through separate coding, given by the Minkowski sum of
$M$ single-subchannel DoF regions, under the total order condition, hence settling the \emph{if} part of the main theorem.
To prove the \emph{only if} part of the theorem, we identify a set of DoF tuples achievable through joint coding across subchannels, yet not achievable through separate coding whenever
the total order condition is violated.
Moreover, we also highlight the implications of our main result on the design of CSIT feedback schemes for multi-carrier multi-antenna
wireless networks.
\end{abstract}
\newpage
\section{Introduction}
Degrees-of-freedom (DoF) studies for wireless networks seek to characterize
the optimal number of interference-free signalling dimensions accessible at each receiver in the
asymptotically high signal to noise ratio (SNR) regime.
While caution must be practiced in translating DoF findings into practical
insights, such findings nevertheless serve as a crude first step along a path of refinements
towards understanding the information-theoretic capacity limits of wireless networks \cite{Davoodi2017}.

A prevalent assumption in initial DoF studies, which mainly focus on interference and multi-antenna wireless networks, was that of the availability of perfect channel state
information at the transmitters (CSIT).
Not long after, however, it became clear that such overly optimistic assumption is difficult to satisfy in practical systems, largely due to the
fading nature of
wireless channels.
This prompted a shift of focus in DoF studies towards incorporating various forms of CSIT imperfections, including: absent instantaneous CSIT \cite{Huang2012,Vaze2012,Rassouli2015}, compound CSIT \cite{Maddah-Ali2010,Gou2011},
finite precision and partial instantaneous CSIT \cite{Davoodi2016,Hao2017a,Davoodi2019},
delayed CSIT \cite{Maddah-Ali2012}, mixed delayed and partial instantaneous CSIT \cite{Gou2012,Yang2013,Chen2013a},
hybrid and alternating CSIT \cite{Tandon2013,Rassouli2016,Lashgari2016}, and topological CSIT \cite{Jafar2014,Yi2015}.
\subsection{MISO BC under Partial CSIT}
\label{subsec:intro_MISO_BC}
As seen through a number of the above-mentioned works, the multiple-input-single-output (MISO) broadcast channel (BC) has been considered a canonical setting for investigating the impact of CSIT inaccuracies on the DoF of wireless networks.
The capacity region (and hence the DoF region) of this channel is well-understood under the
idealistic assumption of perfect CSIT \cite{Weingarten2006}, which in turn provides a
firm starting point for studies that consider more relaxed CSIT assumptions.
Moreover, the earliest observations on the fundamental role of CSIT in interference management were noted through studying this channel \cite{Lapidoth2005,Jindal2006,Weingarten2007}, gaining it a central status in such analysis, although many such early observations were
in the form of conjectures that were settled some years after \cite{Maddah-Ali2010,Gou2011,Davoodi2016}.
The fact that DoF results in the MISO BC constitute outer bounds for more
intricate settings, as the interference channel (IC) and the X channel \cite{Davoodi2017},
also adds to its significance in the development of our understanding of the role of CSIT in wireless networks.

Amongst the various models of CSIT imperfections, the partial instantaneous CSIT model
has become of particular research interest over the recent few years.
Considering the MISO BC under this model, the transmitter is assumed to have access to an erroneous estimate of each user's channel vector, while estimation error terms are assumed to scale as $O(\mathrm{SNR}^{-\alpha_{k}})$,
where $\alpha_{k} \in [0,1]$ is a parameter that captures the CSIT quality level for user $k$.
For instance, $\alpha_{k} = 0$ represents finite precision CSIT, which reduces to no knowledge at the transmitter (N), while $\alpha_{k} = 1$ amounts to perfect channel knowledge (P), both from a DoF viewpoint.
The challenging nature of DoF studies under this CSIT uncertainty model is epitomized by the Lapidoth-Shamai-Wigger conjecture that the sum-DoF of the 2-user MISO BC collapses to $1$ under finite precision CSIT (i.e. $\alpha_{k} = 0$),
which remained open for nearly a decade \cite{Lapidoth2005}.
This conjecture was finally proved by Davoodi and Jafar in a seminal work in which they introduced a novel converse argument named the aligned image sets (AIS) approach \cite{Davoodi2016}.
In particular, Davoodi and Jafar derived an upper bound for the sum-DoF of the $K$-user MISO BC under arbitrary levels of partial CSIT, given by
\begin{equation}
\label{eq:sum_DoF_UB}
d_{\Sigma} \leq  1 + \alpha_{2} + \cdots + \alpha_{K}
\end{equation}
where it is assumed, without loss of generality, that $\alpha_{1} \geq \alpha_{k}$ for all $k$.
It is evident that the DoF collapse to $1$ under finite precision CSIT, as conjectured in \cite{Lapidoth2005}, follows as a special case of the upper bound in \eqref{eq:sum_DoF_UB}.
As for the opposite direction, the achievability of \eqref{eq:sum_DoF_UB} was shown using a scheme based on rate-splitting, with a superposition of zero-forcing and multicasting signals, proposed in \cite{Yang2013} for the 2-user setting and generalized to $K$-user settings in \cite{Clerckx2016} and references therein.

Once equipped with the upper bound in \eqref{eq:sum_DoF_UB}, a polyhedral outer bound for the
entire DoF region is easily constructed by bounding the sum-DoF of each subset of users, while eliminating remaining users.
With this outer bound in hand, the main challenge in going from a sum-DoF characterization to an entire DoF region characterization
becomes the achievability side of the argument.
In particular, the rate-splitting scheme used to achieve the sum-DoF in \eqref{eq:sum_DoF_UB} is, in general,
specified by several design variables for power control and common DoF assignment.
While such design variables can be optimized to obtain a DoF tuple that maximizes a certain
scalar objective function, e.g.  the sum-DoF \cite{Joudeh2016} or the symmetric-DoF \cite{Joudeh2016a},
the entire achievable DoF region is generally described as the collection of DoF tuples achieved
through all combinations of feasible design variables.
Proving achievability hence requires matching the DoF region achieved through rate-splitting, described using a mixture of CSIT parameters and
auxiliary design variables, to the outer bound, expressed in terms of CSIT parameters only.
This was accomplished by Piovano and Clerckx in \cite{Piovano2017} through an exhaustive characterization of all faces describing the outer bound region, and then prescribing tuned strategies that attain all DoF tuples in each such face.

The partial CSIT model described above can be further enriched by allowing CSIT levels to vary not only across users,
but also across signalling dimensions.
By doing so for the MISO BC, we enter the realm of a more intricate and far less understood setting:
the parallel MISO BC under partial CSIT, which is the main focus of this work.
In particular, we consider a $K$-user, $M$-subchannel setting, where transmission occurs over $M$ parallel subchannels,
and we further assume arbitrary levels of partial CSIT for each subchannel $m$, encompassed by the state $(\alpha_{1}^{[m]},\ldots,\alpha_{K}^{[m]}) \in [0,1]^{K}$.
\subsection{Parallel MISO BC under Partial CSIT and Inseparability}
A key issue that arises when studying parallel channel models, which are motivated by fading wireless channels, is \emph{separability}.
This is defined as the optimality of independent coding over subchannels (or fading states), in which each subchannel is treated as a stand-alone network, subject to a joint power constraint across subchannels.
%
Under perfect CSIT, the parallel MISO BC is separable in the strongest sense, i.e.
with respect to its \emph{entire capacity region} \cite{Mohseni2006}.
Separability, however, which also holds for point-to-point and multiple-access channels, turns out to be ``\emph{more of an exception than a rule for wireless networks in general and interference networks in particular}" \cite{Jafar2011}.
As explained in \cite[Ch. 4.4]{Jafar2011}, the main causes for the inseparability of parallel wireless channels in general, revealed
by studying the IC and X channel, are: 1) \emph{antidote links}, which cannot carry desired signals but may deliver signals useful for
interference cancellation, and 2) \emph{interference alignment}, enabled by
alternating network topologies arising from channel state variation.

Under partial CSIT, the parallel MISO BC ceases to be separable in general.
Looking through the DoF lens adopted in this work, it is seen that the luxury of creating non-interfering links through zero-forcing is lost  under CSIT imperfections in general, and as the MISO BC starts to inherit features from the IC and X channel,
the above causes of inseparability come into play.
This is best exemplified by the $2$-user, $2$-subchannel setting with a PN,NP CSIT pattern \cite{Tandon2013},
i.e. CSIT states given by $(1,0)$ and $(0,1)$ as shown in Fig. \ref{fig:topologies}(a).
In \cite{Tandon2013}, Tandon et al. showed that by jointly coding over these $2$ subchannels, a sum-DoF of $3/2$ is achieved, which is also optimal. 
This strictly outperforms separate coding over each state, which achieves a sum-DoF of $1$ at most.\footnote{Interestingly, the fact that separate coding achieves $1$ DoF at most over the CSIT state $(1,0)$, or equivalently $(0,1)$, was mentioned as a conjecture in \cite[Example 3]{Tandon2013}. This was settled in the affirmative later on in \cite{Davoodi2016} (see \eqref{eq:sum_DoF_UB}).}
Other examples of inseparability in settings with more than 2 users and 2 subchannels, or 2-user settings with arbitrary levels of partial CSIT, are given in \cite{Rassouli2016,Hao2013,Chen2013}.

While it is understood that the $K$-user, $M$-subchannel parallel MISO BC is inseparable in general under arbitrary levels of partial CSIT,
a comprehensive understanding of its DoF region is still missing.
Consider the $2$-user, $M$-subchannel special case for instance. An achievable DoF region
attained through joint coding across subchannels is given in \cite{Hao2013},
yet the optimality of such region has not been successfully established.
Furthermore, insights into the daunting complexity
incurred in going beyond 2-user settings, even when restricting to PN-CSIT patterns with $\{0,1\}$ CSIT qualities,
are seen through the examples in \cite{Rassouli2016}.
The formidable nature of this problem of interest, in its generality, motivates a more tractable approach, which we take in this paper.
Instead of striving to characterize the DoF region in the general case, we ask the question of whether there exists a broad regime of parameters in which it is achieved through simple separate coding.
\begin{figure}[t]
\vspace{-1.0cm}
\centering
\includegraphics[width = 0.9\textwidth,trim={2.5cm 12.2cm 2.5cm 0.0cm},clip]{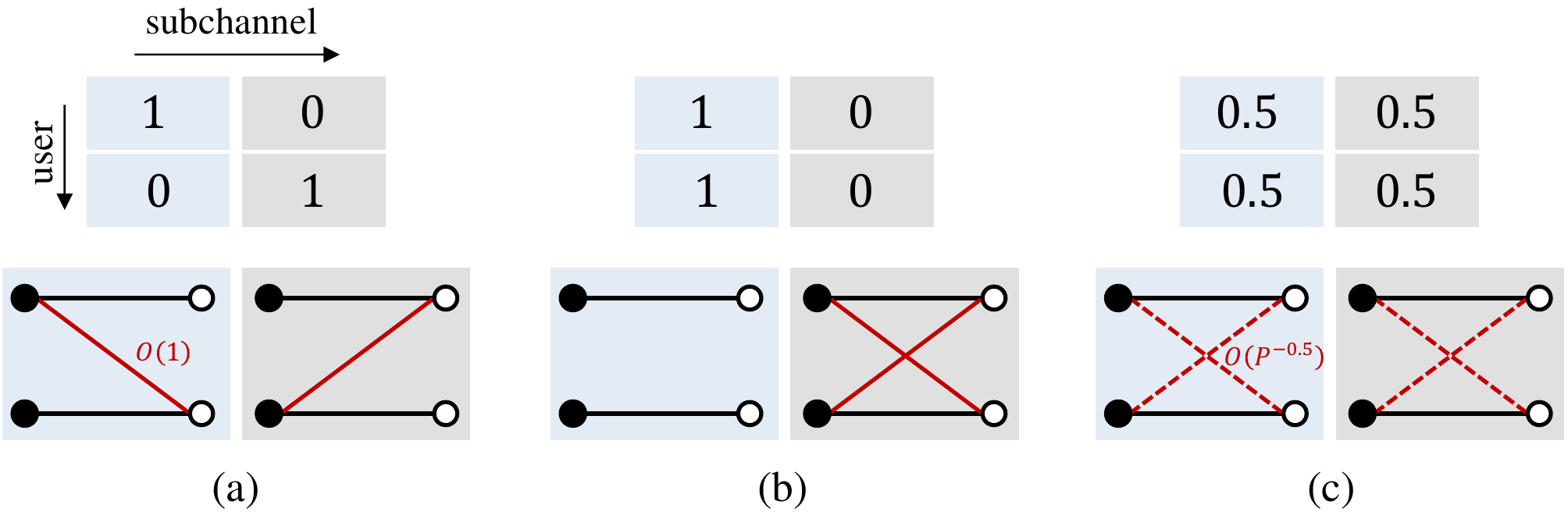}
\caption{\small
Three CSIT patterns for the $2$-user, $2$-subchannel MISO BC (top) and their effective network topologies after
zero-forcing using the possibly inaccurate available CSIT (bottom).
The average CSIT quality for each user is $0.5$ in all settings.
The sum-DoF for all settings is $3/2$,
achieved in (a) using zero-forcing and inter-subchannel common signalling \cite{Tandon2013}, in (b) using zero-forcing in subchannel 1 and common signalling in subchannel 2, and in
(c) using rate-splitting with a superposition of zero-forcing and common signalling in each subchannel \cite{Yang2013}. The setting in (a) is inseparable, while both (b) and (c) are separable.}
\label{fig:topologies}
\end{figure}
\subsection{Overview of Contribution and Organization}
\label{subsec:intro_overview_contrib}
Separability of parallel subchannels is a desirable attribute from a practical standpoint, as it can greatly simplify
coding and multiple access in wireless networks under fading channel conditions.
This desirability of, and need for, simplicity in general has fueled a number of works of late, in
which the optimality of simple coding schemes has been established in broad regimes for various
settings\footnote{Optimality is often established in a DoF, generalized degrees-of-freedom (GDoF), or constant-gap capacity sense, depending on the considered setting and its tractability.} \cite{Geng2015,Sun2016,Joudeh2019b,Yi2018,Davoodi2018a,Joudeh2019a,Chan2019}.
Such simple schemes also tend to be more robust, as they rely less on the fine details of CSIT, making them
all the more attractive for practical purposes.
On the other hand, from a theoretical perspective, the approach of pursuing optimality conditions
for simple schemes has allowed progress on questions which are still open in their generality.
Our pursuit of separability conditions for the parallel MISO BC under partial CSIT can be viewed
in the same spirit as these previous works that focus on studying the optimality of simple schemes.

The main result of this paper is showing that under partial CSIT, the $K$-user, $M$-subchannel parallel MISO BC is separable from an
entire DoF region perspective \emph{if and only if} the corresponding partial CSIT pattern is \emph{totally ordered}.
This condition corresponds to users abiding by the same order, with respect to their CSIT parameters, in each of the parallel subchannels.
In other words, the monotonic order of parameters in a given state $(\alpha_{1}^{[m]},\ldots,\alpha_{K}^{[m]})$ must hold in all other states (see Definition \ref{def:total order} in Section \ref{subsec:partial_CSIT_model}).
In the light of our discussion of inseparability and its causes in the previous subsection, this total order condition for separability
seems natural.
In particular, totally ordered CSIT patterns give rise to parallel subchannels with effectively non-alternating network topologies, see for example
Fig. \ref{fig:topologies}(b) and (c), which in turn, do not provide alignment opportunities arising in alternating topologies, as in Fig. \ref{fig:topologies}(a).
However, despite this intuitive nature of our main result,
showing that it holds is not a straightforward exercise as highlighted in what follows.

The first step towards proving the main result stated above is deriving an outer bound.
To this end, in Section \ref{sec:converse} we extend the AIS approach \cite{Davoodi2016} and derive a sum-DoF upper bound for
the multi-subchannel setting.
In its essence, the AIS approach relies on a combinatorial accounting of the maximum number of codewords that can be aligned
at an undesired receiver while remaining distinguishable at a desired receiver, under partial CSIT.
As transmissions in the considered setting take place over parallel subchannels with arbitrary channel uncertainty levels, the
CSIT quality for any given user may vary across the span of one codeword.
This induces variations in the probability of codeword alignment at undesired receivers,
which in turn, determines the average cardinality of the corresponding aligned image set used to bound the DoF.
By taking such variations into consideration, which is a key difference compared to the proof in
\cite{Davoodi2016}, we arrive at a sum-DoF upper bound expressed similarly to the one in \eqref{eq:sum_DoF_UB},
except that $\alpha_{k} = \frac{1}{M} \sum_{m=1}^{M}\alpha_{k}^{[m]}$ is the average CSIT quality for user $k$ in the multi-subchannel setting.
Interestingly, this sum-DoF upper bound proves the optimality of the $2$-user, $M$-subchannel achievable DoF region derived in \cite{Hao2013}.
More generally, equipped with this sum-DoF upper bound, a polyhedral outer bound for the DoF
region, denoted by $\mathcal{D}_{\mathrm{out}}$, is constructed by bounding the sum-DoF of each subset of users (see Theorem \ref{th:outer_bound}).

Second, the outer bound is employed to show the sufficiency of the total order condition for separability.
In particular, in Section \ref{sec:suff_total_order} we show that for totally ordered CSIT patterns,
we have $\mathcal{D}_{\mathrm{out}} = \frac{1}{M} \big(\mathcal{D}^{[1]} \oplus \mathcal{D}^{[2]} \oplus \cdots \oplus \mathcal{D}^{[M]} \big)$,
where $\mathcal{D}^{[m]}$ is the single-subchannel DoF region for subchannel $m$ when treated as a separate network, and $\oplus$ denotes the Minkowski sum operation.
To this end, we first obtain an equivalent representation of the single-subchannel DoF region
$\mathcal{D}^{[m]}$ in terms of auxiliary variables,
which include DoF and power assignment variables used to tune the achievability scheme.
Interestingly, $\mathcal{D}_{\mathrm{out}}$ assumes a similar equivalent representation, which in turn
enables us to show that for any $\mathbf{d}\in\mathcal{D}_{\mathrm{out}}$, there exists
$\mathbf{d}^{[m]} \in \mathcal{D}^{[m]}$ for each $m$ such that $\mathbf{d} = \frac{1}{M}\sum_{m = 1}^{M} \mathbf{d}^{[m]}$.

As an auxiliary result, in obtaining the equivalent representation of the single-subchannel DoF region $\mathcal{D}^{[m]}$, we provide
an alternative proof for \cite[Th. 1]{Piovano2017}, where the achievability of the single-subchannel DoF region was first established.
This result, given in Lemma \ref{lem:single_subchannel}, is interesting in its own right as it shows that for each subchannel $m$, it is sufficient to optimize only a single power control variable, in addition to the common DoF assignment variables, to achieve all points
of the DoF region $\mathcal{D}^{[m]}$, as opposed to the $K$ power control variables required in  \cite{Piovano2017}.

Third, after establishing the sufficiency of the total order condition for separability, we prove its necessity in Section \ref{sec:Necessity_of_Total_Order}.
This is shown by explicitly characterizing a set of DoF tuples which are achievable through the joint coding scheme proposed for 2-user settings in \cite{Hao2013}, yet are not achievable through separate coding whenever the total order condition is violated.

Some insights are also drawn from the main separability result, as seen in Section \ref{subsec:results_total_order}.
For instance, a direct consequence is that totally ordered CSIT patterns yield maximal DoF regions
under per-user CSIT budget constraints.
Moreover, we show that any parallel MISO BC with a totally ordered partial CSIT pattern can be realized, in the DoF region sense,
by an equivalent parallel MISO BC with a totally ordered  PN-CSIT pattern.
Such observations  provide insights into the design of DoF-optimal CSIT feedback schemes for multi-carrier wireless systems.
\subsection{Notation}
\label{subses:notation}
$a,A$ are scalars, with $A$ often denoting a random variable unless the contrary is obvious from the context.
$\mathbf{a} \triangleq (a_{1},\ldots,a_{k})$ is a $k$-tuple of scalars, which is also considered to be a column vector.
For any subset of indices $\mathcal{S} \subseteq \{1,\ldots,k\}$, we use $\mathbf{a}(\mathcal{S})$ to denote $\sum_{i \in \mathcal{S}}a_{i}$.
A column vector of all ones is denoted by $\mathbf{1}$, with dimension made clear from the context.
$\mathbf{A}$ is a matrix, with dimensions made clear from the context, and
$\mathcal{A}$ is a set.
We use $\mathbb{R}$, $\mathbb{Q}$ and $\mathbb{Z}$ to denote the sets of real, rational and integer numbers, respectively.
For any subset $\mathcal{A} \subseteq \mathbb{R}$ and positive integers $k$ and $m$, $\mathcal{A}^{k \times m}$ is the set of 
all $k \times m$ matrices with entries drawn from $\mathcal{A}$.
For any positive integers $k_{1}$ and $k_{2}$, with $k_{1} \leq k_{2}$,
the sets $\{1,\ldots,k_{1}\}$  and $\{k_{1},\ldots,k_{2}\}$ are denoted by $\langle k_{1} \rangle$ and $\langle k_{1} : k_{2} \rangle$, respectively.
For sets $\mathcal{A}$ and $\mathcal{B}$, $\mathcal{A}\setminus \mathcal{B}$ is the set of elements in
$\mathcal{A}$ and not in $\mathcal{B}$.
For any pair of sets $\mathcal{A},\mathcal{B} \subseteq \mathbb{R}^{k}$, their Minkowski sum $\mathcal{A}  \oplus \mathcal{B}$ is also a set in $\mathbb{R}^{k}$ defined as $\mathcal{A}  \oplus \mathcal{B} \triangleq  \big\{\mathbf{a} + \mathbf{b} : \mathbf{a} \in \mathcal{A},
\ \mathbf{b} \in \mathcal{B} \big\}$.
\section{System Model and Preliminaries}
\label{sec:system_model}
We consider a parallel MISO BC comprising a $K$-antenna transmitter and $K$ single-antenna receivers (users),
in which communication occurs over $M$ parallel subchannels.\footnote{More generally, denoting the number of transmit antennas by $K_{\mathrm{t}}$, results in this paper extend easily to the case  $K_{\mathrm{t}} > K$. 
On the other hand, the case  $K_{\mathrm{t}} < K$ (i.e. overloaded case) incurs additional challenges, requiring insights beyond the results in this paper. Progress on the overloaded case, when $M =1$, has been reported in \cite{Piovano2016}.}
The index sets for receivers (and similarly, transmit antennas) and subchannels are given by  $\mathcal{K} \triangleq \langle K \rangle$ and $\mathcal{M} \triangleq \langle M \rangle$,
respectively.
For transmissions taking place over $n>0$ uses of the parallel channel (e.g. time instances), the input-output relationship
for channel use $t$, where  $ t \in \langle n \rangle$, is given by:
\begin{equation}
\label{eq:received signal}
Y_{k}^{[m]}(t)= \sum_{i \in \mathcal{K}} G_{ki}^{[m]}(t) X_{i}^{[m]}(t)+ Z_k^{[m]}(t), \ m \in \mathcal{M}, k \in \mathcal{K}
\end{equation}
In the above, for channel use $t$ and subchannel $m$, $Y_{k}^{[m]}(t)$ is the signal observed by user $k$, $G_{ki}^{[m]}(t)$ is the fading channel coefficient between transmit antenna $i$ and user $k$, $X_{i}^{[m]}(t)$ is the symbol sent from transmit antenna $i$,
and $Z_k^{[m]}(t) \sim \mathcal{N}_{\mathbb{C}}(0,1)$ is the zero mean unit variance additive white Gaussian noise (AWGN) at user $k$.
All signals and channel coefficients in the above are complex.
The transmitter is subject to the power constraint given by:
\begin{equation}
\label{eq:power_constraint}
\frac{1}{nM} \sum_{t=1}^{n} \sum_{m=1}^{M} \E \left[ | X_{1}^{[m]}(t) |^{2} + | X_{2}^{[m]}(t) |^{2} + \cdots + | X_{K}^{[m]}(t) |^{2} \right]
\leq P
\end{equation}
which can be interpreted as the average transmission power per-channel-use per-subchannel.
\subsection{Partial CSIT}
\label{subsec:partial_CSIT}
Under the partial CSIT model of interest, the channel coefficients associated with user $k$ over subchannel $m$ are modeled as
\begin{equation}
\label{eq:channel model}
G_{ki}^{[m]}(t) = \hat{G}_{ki}^{[m]}(t) + \sqrt{P^{-\alpha_{k}^{[m]}}}\tilde{G}_{ki}^{[m]}(t), \ i \in \mathcal{K}, t  \in \langle n \rangle
\end{equation}
where $\hat{G}_{ki}^{[m]}(t)$ and $\tilde{G}_{k}^{[m]}(t)$ are the corresponding channel estimate and estimation error terms, respectively, while $\alpha_{k}^{[m]} \in [0,1]$ is a CSIT quality level parameter.
We assume non-degenerate channel conditions, where values of all channel variables, alongside the determinants of the
overall channel matrices, are bounded away from zero and infinity \cite{Davoodi2016}.

We also consider a non-degenerate channel uncertainty model, where channel variables $\hat{G}_{ki}^{[m]}(t)$ and $\tilde{G}_{ki}^{[m]}(t)$ are subject to the bounded density assumption\footnote{\label{foot:bounded_density}A set of random variables $\mathcal{G}$ satisfies the bounded density assumption if there existence of a finite positive constant $0<f_{\max}< \infty$ such that for all finite cardinality disjoint subsets $\mathcal{G}_{1},\mathcal{G}_{2} \subset \mathcal{G}$, the join probability density function of variables in $\mathcal{G}_{1}$ conditioned on variables in $\mathcal{G}_{2}$ exists and is bounded above by $f_{\max}^{|\mathcal{G}_{1}|}$.} (see \cite[Sec. II.D]{Davoodi2016} and
\cite[Definition 1]{Davoodi2018}).
The main difference between $\hat{G}_{ki}^{[m]}(t)$ and $\tilde{G}_{ki}^{[m]}(t)$ is that the actual realizations of the former
are revealed to the transmitter, while the realizations of the latter remain unknown to the transmitter.
Under this CSIT uncertainty model, the parameter $\alpha_{k}^{[m]} \in [0,1]$  captures the whole range of knowledge available at the transmitter of user $k$'s channel coefficients over subchannel $m$, i.e.
$\alpha_{k}^{[m]} = 0$ essentially reduces to the case where channel knowledge is absent, while $\alpha_{k}^{[m]} = 1$
amounts to perfectly known CSIT, both in the DoF sense.
\subsection{CSIT Pattern and Total Order}
\label{subsec:partial_CSIT_model}
We define the \emph{CSIT pattern} $\mathbf{A} \in [0,1]^{K \times M}$ as the matrix of CSIT parameters given by
\begin{equation}
 \mathbf{A} \triangleq \left[
                         \begin{array}{cccc}
                           \alpha_{1}^{[1]} &  \alpha_{1}^{[2]} & \cdots & \alpha_{1}^{[M]} \\
                           \alpha_{2}^{[1]} &  \alpha_{2}^{[2]} & \cdots & \alpha_{2}^{[M]} \\
                           \vdots & \vdots  & \ddots & \vdots \\
                           \alpha_{K}^{[1]} &  \alpha_{K}^{[2]} & \cdots & \alpha_{K}^{[M]} \\
                         \end{array}
                       \right].
\end{equation}
In scenarios where CSIT is either perfect or not available for all users, we have $ \mathbf{A}  \in \{0,1\}^{K \times M} $,
which we refer to as a PN-CSIT pattern.
The CSIT state for subchannel $m$ is given by the tuple $\bm{\alpha}^{[m]} \triangleq ( \alpha_{1}^{[m]} ,\ldots,\alpha_{K}^{[m]})$,
formed by the CSIT parameters for all $K$ users over subchannel $m$.
On the other hand, the CSIT tuple associated with user $k$ over all $M$ subchannels is given by
$\bm{\alpha}_{k} \triangleq  (\alpha_{k}^{[1]} ,\ldots,\alpha_{k}^{[M]} )$.
Recalling that we take tuples to represent column vectors, the CSIT pattern $\mathbf{A}$ is compactly written as
\begin{equation}
 \mathbf{A}  = \left[\begin{array}{cccc}
 \bm{\alpha}^{[1]} &  \bm{\alpha}^{[2]} & \cdots & \bm{\alpha}^{[M]}
 \end{array}
 \right]
  = \left[\begin{array}{cccc}
 \bm{\alpha}_{1} &  \bm{\alpha}_{2} & \cdots & \bm{\alpha}_{K}
 \end{array}
 \right]^{\Trn}
\end{equation}
The \emph{average CSIT quality} for user $k$ is defined as
\begin{equation}
\alpha_{k} \triangleq \frac{1}{M} \sum_{m=1}^{M} \alpha_{k}^{[m]}
\end{equation}
from which we construct the average CSIT state as
\begin{equation}
\bm{\alpha} \triangleq (\alpha_{1} ,\ldots,\alpha_{K} ) = \frac{1}{M}\mathbf{A} \cdot \mathbf{1}.
\end{equation}
Without loss of generality, we may assume the following order of average CSIT qualities:
 \begin{equation}
\label{eq:average CSIT order}
\alpha_{1} \geq \alpha_{2} \geq \cdots \geq \alpha_{K}.
\end{equation}
Next, we introduce the  notion of total order.
\begin{definition}
\label{def:total order}
Users are totally ordered with respect to CSIT parameters if there exists a permutation $\pi$ over $\langle K \rangle$  such that $\bm{\alpha}_{\pi(1)} \geq \bm{\alpha}_{\pi(2)} \geq \cdots \geq \bm{\alpha}_{\pi(K)}$, where the vector inequalities are element-wise.
Under the average CSIT order in \eqref{eq:average CSIT order}, the condition for total order becomes
\begin{equation}
\label{eq:total_CSIT_order}
\bm{\alpha}_{1} \geq \bm{\alpha}_{2} \geq \cdots \geq \bm{\alpha}_{K}.
\end{equation}
\end{definition}
According to the above definition, and assuming that \eqref{eq:average CSIT order} always holds without loss of generality,
the entries of each column $\bm{\alpha}^{[m]}$
of a totally ordered CSIT pattern $\mathbf{A}$
are non-increasing with respect to the user index $k$.
Moreover, for fixed $K$ and $M$ which specify a corresponding class of parallel MISO BCs, we use  $\mathcal{A}_{\mathrm{to}}$, where $\mathcal{A}_{\mathrm{to}} \subseteq [0,1]^{K \times M}$,  to denote the corresponding set of all CSIT patterns $\mathbf{A}$ which are totally ordered according to \eqref{eq:total_CSIT_order}.
\subsection{Messages, Rates, Capacity and DoF}
\label{subsec:msgs-rates-capacity-DoF}
The transmitter has messages $W_1,\ldots, W_K$ intended to users $1,\ldots,K$, respectively.
Achievable rate tuples $(R_1(P),\ldots,R_K(P))$ and the capacity region $\mathcal{C}(P)$ are all defined in the standard Shannon theoretic sense.
Note that achievable rates are defined as $n \rightarrow \infty$, yet $M$ remains fixed for a given channel.
The DoF tuple $\mathbf{d} \triangleq (d_{1}, \ldots, d_{K})$ is said to be achievable if there exists $(R_1(P),\ldots,R_K(P)) \in \mathcal{C}(P)$ such that $d_k=\lim_{P \to \infty} \frac{R_k(P)}{M\log(P)}$ for all $k \in \langle K \rangle$.
Here $M\log(P)$ approximates the baseline capacity of the $M$ subchannels at high SNR.
The DoF region is denoted by $\mathcal{D}$, and is defined as the closure of all achievable DoF tuples $\mathbf{d}$.
As the setting of interest is parameterized by the CSIT pattern $\mathbf{A}$,
we occasionally make this dependency explicit  in the DoF region, i.e. $\mathcal{D}(\mathbf{A})$, especially when
comparing channels with different CSIT patterns.
\begin{remark}
\label{remark:DoF per subchannel}
According to the above definition, the considered DoF is per-channel-use per-subchannel.
For example, if channel uses and subchannels represent time instances and orthogonal frequency sub-carriers respectively,
the DoF represents the number of interference free spatial signalling dimensions per orthogonal time-frequency signalling dimension at high SNR.
\end{remark}
\subsection{Separate Coding and Separability}
To set the stage for our main result, presented in the following section, we here present an inner bound for $\mathcal{D}$ achieved through
separate coding over each subchannel.
First, let us consider subchannel $m$, where $m \in \mathcal{M}$, as a stand-alone network and let us denote its optimal
DoF region by $\mathcal{D}^{[m]}$, which consists of all achievable DoF tuples $\mathbf{d}^{[m]} = \big(d_{1}^{[m]},\ldots,d_{K}^{[m]}\big)$
over subchannel $m$.
From \cite[Th. 1]{Piovano2017}, we know that the DoF region $\mathcal{D}^{[m]}$
is given by
\begin{equation}
\label{eq:DoF_region_m}
\mathcal{D}^{[m]} \triangleq \Big\{ \mathbf{d}^{[m]} \in \mathbb{R}_{+}^{K} : \mathbf{d}^{[m]}(\mathcal{S}) \leq 1  +
\bm{\alpha}^{[m]}(\mathcal{S}) - \max_{j \in \mathcal{S}} \alpha_{j}^{[m]}, \ \mathcal{S} \subseteq \mathcal{K} \Big\}.
\end{equation}
Going back to the MISO BC with $M$ parallel subchannels,
separate coding can be carried out over each subchannel to achieve
any DoF tuple $\mathbf{d} = (d_{1},\ldots,d_{K}) \in \mathcal{D}$ of the form
\begin{equation}
\mathbf{d} = \frac{1}{M}\sum_{m \in \mathcal{M}}\mathbf{d}^{[m]},
\text{ for some } \mathbf{d} ^{[m]} \in \mathcal{D}^{[m]},  m \in \mathcal{M}.
\end{equation}
This separation based approach results in an achievable DoF region given by
\begin{equation}
\label{eq:DoF_region_in}
\mathcal{D}_{\mathrm{sep}} \triangleq \frac{1}{M} \bigoplus_{m\in \mathcal{M}} \mathcal{D}^{[m]}  = \frac{1}{M} \big(\mathcal{D}^{[1]} \oplus \mathcal{D}^{[2]} \oplus \cdots \oplus \mathcal{D}^{[M]} \big)
\end{equation}
where we recall that $\oplus $ is the Minkowski sum operation (see Section \ref{subses:notation}).
As $\mathcal{D}_{\mathrm{sep}}$ is achievable, it readily follows that $\mathcal{D}_{\mathrm{sep}} \subseteq \mathcal{D}$.
Separability, in a DoF region sense, holds  when $\mathcal{D}_{\mathrm{sep}}$ and $\mathcal{D}$ coincide.
\begin{definition}
\label{def:separability}
The MISO BC with parallel subchannels is separable from a DoF region perspective if and only if
$\mathcal{D} = \mathcal{D}_{\mathrm{sep}} =
\frac{1}{M} \big(\mathcal{D}^{[1]} \oplus \mathcal{D}^{[2]} \oplus \cdots \oplus \mathcal{D}^{[M]} \big) $.
\end{definition}
Before concluding this section, we present a definition for the parameterized class of
polyhedra that encompass the region in \eqref{eq:DoF_region_m}.
As we see in consequent parts, this class of polyhedra plays a central role in our DoF region characterizations and separability result.
\begin{definition}
\label{def:pol_region}
Given a $K$-tuple of parameters $\bm{\beta} \in [0,1]^{K}$, the polyhedron  $\mathcal{P}(\bm{\beta})$
is defined as
\begin{equation}
\label{eq:pol_region}
\mathcal{P}(\bm{\beta}) \triangleq \Big\{ \mathbf{d} \in \mathbb{R}_{+}^{K} : \mathbf{d}(\mathcal{S}) \leq 1  +
\bm{\beta}(\mathcal{S}) - \max_{j \in \mathcal{S}} \beta_{j}, \ \mathcal{S} \subseteq \mathcal{K} \Big\}.
\end{equation}
\end{definition}
It is easy to verify that for each subchannel $m \in \mathcal{M}$, we have $\mathcal{D}^{[m]} = \mathcal{P}(\bm{\alpha}^{[m]})$.
\section{Main Results and Insights}
\label{sec:main_results}
In this section, we present the main results of this work alongside some observations and insights.
\begin{figure}[t]
\vspace{-1.0cm}
\centering
\includegraphics[width = 0.9\textwidth,trim={3.76cm 11.6cm 3.76cm 1.0cm},clip]{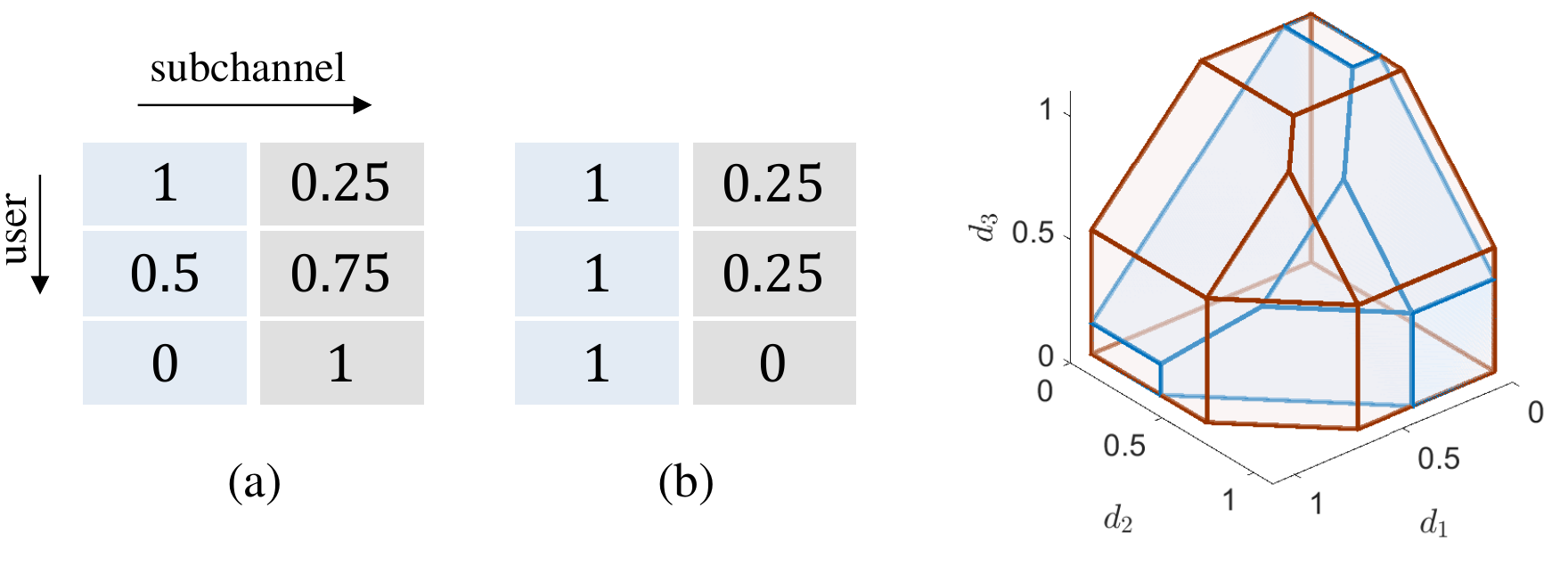}
\caption{\small
Two distinct CSIT patterns, $\mathbf{A}$ in (a) and $\mathbf{B}$  in (b), for the $3$-user, $2$-subchannel setting that have the same average CSIT state, i.e. $\bm{\alpha} = \frac{1}{2}\mathbf{A} \cdot \mathbf{1} = \frac{1}{2}\mathbf{B} \cdot \mathbf{1} = (0.625,0.625,0.5)$.
On the right-hand-side, we have the outer bound DoF region $\mathcal{D}_{\mathrm{out}} = \mathcal{P}(\bm{\alpha})$ in red, superimposed on top of the separate coding achievable DoF region $\mathcal{D}_{\mathrm{sep}}(\mathbf{A})$ for the CSIT pattern $\mathbf{A}$ in blue. The separate coding achievable DoF region for the CSIT pattern $\mathbf{B}$ coincides with the outer bound, i.e. $\mathcal{D}_{\mathrm{sep}}(\mathbf{B}) = \mathcal{P}(\bm{\alpha})$.}
\label{fig:regions}
\end{figure}
\subsection{Outer Bound}
\label{subsec:outer_bound}
We start by presenting an outer bound for the DoF region $\mathcal{D}$,
whose proof is given in Section \ref{sec:converse}.
\begin{theorem}
\label{th:outer_bound}
For the parallel MISO BC under partial CSIT described in Section \ref{sec:system_model},
the DoF region $\mathcal{D}$ is included in the polyhedral region $\mathcal{D}_{\mathrm{out}}$ given by
\begin{equation}
\label{eq:outer_bound_region}
\mathcal{D}_{\mathrm{out}} \triangleq \Big\{ \mathbf{d} \in \mathbb{R}_{+}^{K} : \mathbf{d}(\mathcal{S}) \leq 1 + \bm{\alpha}\big(\mathcal{S} \setminus \{\min \mathcal{S}\} \big), \ \mathcal{S} \subseteq \mathcal{K} \Big\}.
\end{equation}
\end{theorem}
It is worthwhile highlighting that due to the order of average CSIT qualities in \eqref{eq:average CSIT order},
for any $\mathcal{S} \subseteq \mathcal{K}$, the corresponding sum-DoF inequality in
\eqref{eq:outer_bound_region} is equivalently expressed as
\begin{equation}
\mathbf{d}(\mathcal{S}) \leq 1 + \bm{\alpha}(\mathcal{S} ) - \max_{i \in \mathcal{S}}\{ \alpha_{i}\}.
\end{equation}
It readily follows that the outer bound in \eqref{eq:outer_bound_region} belongs to the class of polyhedra in Definition \ref{def:pol_region},
i.e. $\mathcal{D}_{\mathrm{out}} = \mathcal{P}(\bm{\alpha})$.
It is also evident that this outer bound depends only on the average CSIT state $\bm{\alpha}$,
and hence does not distinguish between different CSIT patterns that have the same average CSIT qualities, e.g.
distinct $\mathbf{A}$  and $\mathbf{B}$ in $[0,1]^{K \times M}$ with $\mathbf{A} \cdot \mathbf{1} = \mathbf{B} \cdot \mathbf{1}$ (see Fig. \ref{fig:regions}).
An implication of overlooking the details of CSIT patterns is that $\mathcal{D}_{\mathrm{out}}$ is not tight in general.
This is further elaborated through the following observations, in which we examine $\mathcal{D}_{\mathrm{out}}$ in
light of prior results.
\begin{enumerate}
\item \textbf{Single Subchannel:} It can be easily verified that in the single-subchannel case, i.e. $M = 1$,
the outer bound in \eqref{eq:outer_bound_region} reduces to the DoF region characterized in \cite{Piovano2017} (see \eqref{eq:DoF_region_m}).
\item \textbf{Two Users:} For the 2-user case with an arbitrary number of subchannels $M$, the outer bound in Theorem
\ref{th:outer_bound} boils down to
\begin{equation}
\label{eq:DoF_region_2_users}
\big\{ (d_{1},d_{2}) \in \mathbb{R}_{+}^{2} : d_{1} \leq 1, \ d_{2} \leq 1, \ d_{1}+d_{2} \leq 1 + \alpha_{2} \big\}.
\end{equation}
The region in \eqref{eq:DoF_region_2_users} was shown to be achievable in \cite{Hao2013} using a scheme that
performs joint coding across the $M$ subchannels, in general.
The optimality of this achievable region, however, remained open.\footnote{An attempt to prove the optimality of \eqref{eq:DoF_region_2_users}
in \cite{Hao2013} was shown to be flawed in \cite{Davoodi2016}.}
The result in Theorem \ref{th:outer_bound} settles this issue.
\item \textbf{Beyond Two Users:} Beyond the 2-user or single-subchannel cases,
the outer bound in Theorem \ref{th:outer_bound} is not tight in general.
This can be inferred from the 3-user, 3-subchannel example in \cite[Fig. 3]{Rassouli2016}, with a symmetric PN-CSIT pattern formed by the states $(1,0,0)$, $(0,1,0)$ and $(0,0,1)$, through which it is shown that average CSIT qualities\footnote{Equivalent to  the marginal prababilities, or \emph{marginals}, in the the alternating CSIT context considered in \cite{Rassouli2016}.} on their own
are generally insufficient to describe tight DoF outer bounds for the MISO BC with parallel subchannels whenever $K \geq 3$.
Alternatively, tighter outer bounds are derived by taking into account the specific, and possibly alternating (i.e. non-totally ordered), structure of CSIT patterns, which  is not captured by the corresponding average CSIT states \cite[Sec. V]{Rassouli2016}.
\end{enumerate}
\subsection{Total Order and Separability}
\label{subsec:results_total_order}
Although not tight in general, the outer bound presented in Theorem \ref{th:outer_bound} is in fact entirely achievable
for a broad regime of CSIT patterns identified in the following result (see, e.g., Fig. \ref{fig:regions}).
\begin{theorem}
\label{th:total_order}
The parallel MISO BC under partial CSIT described in Section \ref{sec:system_model}
is separable from a DoF region perspective if and only if the corresponding CSIT pattern is totally ordered, i.e. $\mathbf{A} \in \mathcal{A}_{\mathrm{to}}$.
Moreover, under this total order condition, the DoF region is given by
\begin{equation}
\label{eq:th_DoF_region_total_order}
   \mathcal{D} = \mathcal{D}_{\mathrm{sep}} =  \mathcal{D}_{\mathrm{out}}.
\end{equation}
\end{theorem}
The sufficiency of the total order condition for separability is proved in Section \ref{sec:suff_total_order} by
showing that $\mathcal{D}_{\mathrm{sep}} = \mathcal{D}_{\mathrm{out}}$ whenever $\mathbf{A} \in \mathcal{A}_{\mathrm{to}}$.
On the other hand, the necessity of the total order condition for separability is proved in Section \ref{sec:Necessity_of_Total_Order},
where we show that $\mathcal{D}_{\mathrm{sep}} \subset \mathcal{D}$ whenever $\mathbf{A} \notin \mathcal{A}_{\mathrm{to}}$.
In what follows, we draw some insights from the results in Theorem \ref{th:outer_bound} and Theorem \ref{th:total_order}.
\begin{enumerate}
\item \textbf{Maximal DoF region under CSIT budget constraints:}
The above results offer insights into the optimal allocation
of CSIT resources across subchannels under \emph{per-user} budget constraints.
In particular, consider a constraint on CSIT budgets given by
\begin{equation}
\frac{1}{M} \mathbf{A} \cdot \mathbf{1} = \bm{\alpha} \leq \bm{\alpha}^{\star},
\end{equation}
where $\bm{\alpha}^{\star} = (\alpha_{1}^{\star},\ldots,\alpha_{K}^{\star})$ is
the tuple of maximum affordable average CSIT qualities.
Under such constraint, we know from Theorem \ref{th:outer_bound} that any admissible CSIT pattern $\mathbf{A}$ gives rise to a DoF region contained in $\mathcal{D}_{\mathrm{out}} =  \mathcal{P}(\bm{\alpha}^{\star})$.
Moreover, Theorem \ref{th:total_order} tells us that this maximal DoF region, given by $\mathcal{P}(\bm{\alpha}^{\star})$, is attainable with
separate coding over subchannels whenever $\frac{1}{M} \mathbf{A} \cdot \mathbf{1} = \bm{\alpha}^{\star}$
and $\mathbf{A} \in \mathcal{A}_{\mathrm{to}}$.
Therefore, in scenarios where CSIT patterns can be controlled through, for example, flexible allocation of uplink feedback resources,
abiding by the total order condition not only simplifies coding, but also yields maximal DoF regions.
\item \textbf{Only average CSIT qualities matter under total order:} 
The above results  imply that under total order, the DoF region
is described using average CSIT qualities only.
That is 
\begin{equation}
\label{eq:DoF_region_total_order_av_CSIT_only}
\mathcal{D}(\mathbf{A}) = \mathcal{P}(\bm{\alpha}), \ \text{for all} \ \mathbf{A} \in \mathcal{A}_{\mathrm{to}} \ \text{such that} \ \frac{1}{M} \mathbf{A} \cdot \mathbf{1} = \bm{\alpha}
\end{equation}
This observation may have operational significance, as some CSIT patters are more favourable than others in terms of implementation. 
This is explored further in the following point.
\item \textbf{PN-Decomposition:}
\begin{figure}[t]
\vspace{-1.0cm}
\centering
\includegraphics[width = 0.55\textwidth,trim={6.9cm 14cm 6.9cm 0cm},clip]{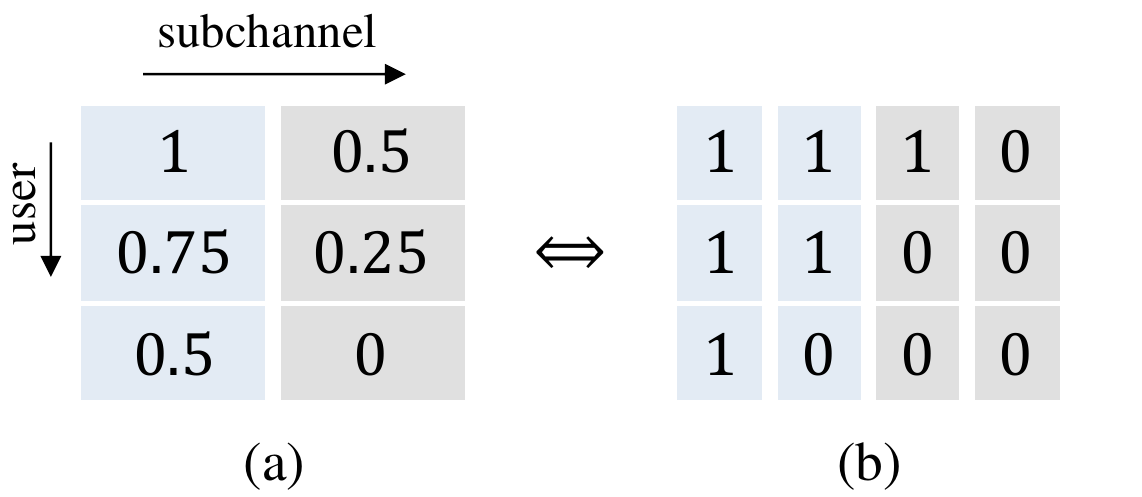}
\caption{\small
Two distinct CSIT patterns, $\mathbf{A}$ in (a) and $\mathbf{A'}$  in (b), with the same average CSIT state given by $\bm{\alpha} = (0.75,0.5,0.25)$. As both $\mathbf{A}$ and $\mathbf{A'}$ are totally ordered, we have $\mathcal{D}(\mathbf{A}) = \mathcal{D}(\mathbf{A}') = \mathcal{P}(\bm{\alpha})$.}
\label{fig:PN_decomp}
\end{figure}
A consequence of the above point is that for any parallel MISO BC with a totally ordered partial CSIT pattern $\mathbf{A}$ of \emph{rational} entries, i.e. 
$\mathbf{A} \in \mathcal{A}_{\mathrm{to}} \cap \mathbb{Q}^{K \times M}$, there is an equivalent parallel MISO BC with some totally ordered PN-CSIT pattern $\mathbf{A}'$, 
i.e. $\mathbf{A}'  \in \mathcal{A}_{\mathrm{to}} \cap \{0,1 \}^{K \times M}$. Equivalence here is in the DoF region sense, i.e. 
$\mathcal{D}(\mathbf{A}) = \mathcal{D}(\mathbf{A}')$.

To show the above, let us define $\mathbf{p}_{l}$ as the PN-CSIT state in which CSIT is perfect for the first $l$ users, where $l \in \langle 0 : K \rangle$, and not available for  the remaining $K-l$ users, that is
\begin{equation}
\label{eq:p_l}
\mathbf{p}_{l} \triangleq (\underbrace{1,\ldots,1}_{l \ \text{entries}},\underbrace{0,\ldots,0}_{K-l \ \text{entries}}), \ l \in \langle 0 : K \rangle.
\end{equation}
Consider an arbitrary CSIT pattern $\mathbf{A}$ with corresponding average CSIT state $\bm{\alpha}$.
We observe that $\bm{\alpha}$ can be decomposed as a weighted-sum
of the above PN-CSIT states as
\begin{equation}
\label{eq:w_CSIT_decomp}
\bm{\alpha} = \sum_{l = 0}^{K} w_{l} \mathbf{p}_{l},
\end{equation}
where the $l$-th weight is defined as $w_{l} \triangleq \alpha_{l} - \alpha_{l+1}$, such that
$\alpha_{0} = 1$ and $\alpha_{K+1} = 0$.
The above weights are  non-negative due to the order in \eqref{eq:average CSIT order}, and they clearly satisfy  $\sum_{l = 0}^{K} w_{l} = 1$.
Moreover, if we further assume that $\mathbf{A} \in \mathbb{Q}^{K \times M}$, then the average CSIT qualities $\alpha_{1},\ldots,\alpha_{K}$ are all rational, which in turn 
implies that the weights $w_{0},w_{1},\ldots,w_{K}$ are all rational as well. In this case, $\bm{\alpha}$ in \eqref{eq:w_CSIT_decomp} may be expressed as
\begin{equation}
\label{eq:av_CSIT_decomp}
\bm{\alpha} = \frac{1}{M'}\sum_{m' = 1}^{M'} \bm{\alpha}'^{[m']},
\end{equation}
for some integer $M'$ and PN-CSIT states $\bm{\alpha}'^{[m']} \in \{\mathbf{p}_{l} : l \in \langle 0:K \rangle \}$, for all  $m' \in \langle M'\rangle$,
of which some are possibly replicas of others.\footnote{For instance, assuming that $w_{l} = N_{l}/M_{l}$ for some integers $N_{l}$ and $M_{l}$, we may take $M' = M_{0} \times M_{1} \times \cdots \times M_{K}$ in \eqref{eq:av_CSIT_decomp}. In this case, each PN-CSIT state $\mathbf{p}_{l}$
is replicated $N_{l}M' / M_{l}$ times, where $N_{l}M' / M_{l}$ is an integer.} 
We take $\mathbf{A}'$ to be the PN-CSIT pattern composed of such states and $\mathcal{D}(\mathbf{A}')$ to be the corresponding DoF region.
It is evident that $\mathbf{A}' \in \mathcal{A}_{\mathrm{to}} \cap \{0,1 \}^{K \times M}$
and $\frac{1}{M} \mathbf{A} \cdot \mathbf{1} = \frac{1}{M} \mathbf{A}' \cdot \mathbf{1} = \bm{\alpha}$.

Now if we further assume that $\mathbf{A} \in \mathcal{A}_{\mathrm{to}}$, then it follows from Theorem \ref{th:total_order} and
 \eqref{eq:DoF_region_total_order_av_CSIT_only} that
\begin{equation}
\label{eq:PN_decomp_eq}
\mathcal{D}(\mathbf{A}) = \frac{1}{M} \bigoplus_{m=1}^{M} \mathcal{P}\left(\bm{\alpha}^{[m]}\right) = \mathcal{P}(\bm{\alpha}) = \frac{1}{M'} \bigoplus_{m=1}^{M'} \mathcal{P}\left(\bm{\alpha}'^{[m']}\right) = \mathcal{D}(\mathbf{A}').
\end{equation}
We further note that \eqref{eq:PN_decomp_eq} can be expressed as the following weighted Minkowski sum:
\begin{equation}
\label{eq:PN_decomp_in}
\mathcal{D}(\mathbf{A})   =
w_{0}\mathcal{P}\left( \mathbf{p}_{0} \right) \oplus
w_{1} \mathcal{P}\left( \mathbf{p}_{1} \right)  \oplus \cdots \oplus
w_{K} \mathcal{P}\left( \mathbf{p}_{K} \right).
\end{equation}
Recalling that $\mathcal{P}\left( \mathbf{p}_{l} \right)$ is the DoF region for a single subchannel with a PN-state $\mathbf{p}_{l}$,
it follows that the weight $w_{l}$ in \eqref{eq:PN_decomp_in} may be interpreted as the fraction of signalling
dimensions in which the effective CSIT state is $\mathbf{p}_{l}$, i.e. perfect CSIT for the first $l$ users and no CSIT for the remaining $K-l$ users.
Moreover, the average CSIT quality $\alpha_{k} = \sum_{l = k}^{K}w_{l}$
may be interpreted as the fraction of signallings dimensions in which perfect CSIT is available for user $k$.
This PN-decomposition interpretation is inline with, and extends, the weighted-sum interpretation in \cite{Hao2013,Hao2017}
and the notion of signal-space partitioning in \cite{Yuan2016,Davoodi2018}.

From the above, it follows that under total order and taking the viewpoint of user $k$,
reporting partial CSIT with average quality $\alpha_{k}$ is equivalent to
reporting perfect CSIT over a fraction $\alpha_{k}$ of subchannels, and no CSIT over the remaining subchannels.
An illustrative example is shown in Fig. \ref{fig:PN_decomp}.
This may have an operational significance for example in OFDMA
systems where CSIT feedback is carried out over a fraction of sub-carriers only.
\end{enumerate}
\begin{remark}
The above PN-decomposition is a special case of a more general decomposition, which we refer to as the Q-PN-decomposition, where Q-P stands for quasi-perfect.
In this more general decomposition, we obtain an equivalent Q-PN-CSIT pattern with entries drawn from  $\{0\} \cup [1-\epsilon_{\mathrm{Q}},1]$, where $\epsilon_{\mathrm{Q}} \geq 0$ is some Q-P tolerance parameter. 
That is, CSIT is either not available or quasi-perfect, where the latter corresponds to a quality parameter in $[1-\epsilon_{\mathrm{Q}},1]$.
This may be useful in scenarios where obtaining perfect CSIT is not possible, or where the original CSIT pattern we wish to decompose 
has irrational entries.
For instance, consider a 2-user, 3-subchannel setting with CSIT states
 $\bm{\alpha}^{[1]} = (0.2\pi,0.2\pi)$, $\bm{\alpha}^{[2]} = (0.2\pi,0.1 \pi)$ and $\bm{\alpha}^{[3]} = (0.2\pi,0)$.
This totally ordered CSIT pattern is equivalent to the Q-PN-CSIT pattern with states 
$\bm{\alpha}'^{[1]} = (0.3\pi,0.3\pi)$, $\bm{\alpha}'^{[2]} = (0.3\pi,0)$ and $\bm{\alpha}'^{[3]} = (0,0)$,
where CSIT is either not available or quasi-perfect with tolerance $\epsilon_{\mathrm{Q}} \approx 0.058$.
\end{remark}
\section{Proof of Outer Bound}
\label{sec:converse}
In this section, we present a proof for the outer bound in Theorem \ref{th:outer_bound}.
As single-user bounds in \eqref{eq:outer_bound_region} are trivial, we focus on the sum-DoF bound given by
\begin{equation}
\label{eq:sum_DoF_bound}
\mathbf{d}(\mathcal{K}) \leq 1 + \bm{\alpha}\big(\mathcal{K} \setminus \{\min \mathcal{K}\} \big).
\end{equation}
All remaining multi-user bounds in \eqref{eq:outer_bound_region}, corresponding to subsets $\mathcal{S} \subset \mathcal{K}$, are derived in a similar manner after eliminating users in $\mathcal{K}\setminus \mathcal{S}$ and their corresponding messages.
The proof of the bound in \eqref{eq:sum_DoF_bound} is based on the AIS approach and follows in the footsteps of the proof for the single-subchannel case in \cite{Davoodi2016}, with modifications to accommodate for the multi-subchannel setting, as alluded to in Section \ref{subsec:intro_overview_contrib}.
For ease of exposition, we focus on real channels. The extension to complex channels is notationally cumbersome, yet conceptually straightforward as demonstrated and noted in \cite{Davoodi2016,Davoodi2018}.

The first step is to reduce the number of channel coefficients through a canonical transformation, performed here for each subchannel.
This yields the partially connected channel model given by
\begin{equation}
\label{eq:received signal_2}
Y_{k}^{[m]}(t)= X_{k}^{[m]}(t) + \sum_{i = 1}^{k-1} G_{ki}^{[m]}(t) X_{i}^{[m]}(t)+ Z_k^{[m]}(t), \ m \in \mathcal{M}, k \in \mathcal{K}
\end{equation}
As shown in \cite{Davoodi2016} for the single-subchannel case,  the above transformation is enabled by the non-degenerate channel model assumption described in Section \ref{subsec:partial_CSIT}, where all values are bounded away from zero and infinity. It can be shown that the canonical channel in \eqref{eq:received signal_2} and the original channel in  \eqref{eq:received signal} have the same DoF (see the appendix of \cite{Davoodi2016} for more details).\footnote{Strictly speaking, in going from the original model in \eqref{eq:received signal} to the canonical model in \eqref{eq:received signal_2}, we obtain new channel coefficients and a new power constraint, which scales as $O(P)$.
With a slight abuse of notation and with no influence on the DoF result, however, we maintain the notation of the original model in this section.
Moreover, note that the partial CSIT model in \eqref{eq:channel model} is inherited by the canonical channel model in
\eqref{eq:received signal_2}.}

The next step is to convert the channel in \eqref{eq:received signal_2} into a deterministic equivalent channel with inputs and outputs all being integers \cite{Avestimehr2015}.
The equivalent deterministic channel is given by
\begin{equation}
\label{eq:received signal_det}
\bar{Y}_{k}^{[m]}(t) = \bar{X}_{k}^{[m]}(t) + \sum_{i = 1}^{k-1} \bigl\lfloor G_{ki}^{[m]}(t) \bar{X}_{i}^{[m]}(t) \bigr\rfloor, \ m \in \mathcal{M}, k \in \mathcal{K}
\end{equation}
where  $\bar{X}_{k}^{[m]}(t) \in \big\{0,\ldots, \lceil \sqrt{P} \rceil \big\}$ and $\bar{Y}_{k}^{[m]}(t) \in \mathbb{Z}$ are the corresponding inputs and outputs respectively.
As shown in \cite[Lem. 1]{Davoodi2016}, the above deterministic approximation has no influence on the DoF.
Therefore, we focus on the channel in \eqref{eq:received signal_det} henceforth.

In what follows, we use $\bar{P}$ to denote $\sqrt{P}$.
We also use $\mathcal{G}_{k}$ to denote the set of channel variables associated with user $k$, and $\mathcal{G}$ to denote the set of all channel variables, i.e.
\begin{equation}
\nonumber
\mathcal{G}_{k} \triangleq \bigl\{  G_{ki}^{[m]}(t) : i \in \langle k-1 \rangle, t \in \langle n \rangle, m \in \mathcal{M}  \bigr\}
\quad \text{and} \quad
\mathcal{G} \triangleq \bigl\{\mathcal{G}_{k} : k \in \mathcal{K} \bigr\}.
\end{equation}
Moreover, we use $\bar{X}_{k}^{[m]}$ with a suppressed time index to denote the sequence $\big( \bar{X}_{k}^{[m]}(1),\ldots, \bar{X}_{k}^{[m]}(n) \big)$, and $\bar{X}_{k}$ with a suppressed subchannel index to denote $\big(  \bar{X}_{k}^{[1]},\ldots,\bar{X}_{k}^{[M]} \big)$.
Similarly, we use $\bar{Y}_{k}^{[m]}$  and $\bar{Y}_{k}$ to denote $\big( \bar{Y}_{k}^{[m]}(1),\ldots, \bar{Y}_{k}^{[m]}(n) \big)$ and $\big(  \bar{Y}_{k}^{[1]},\ldots,\bar{Y}_{k}^{[M]} \big)$, respectively.

We proceed by invoking Fano's inequality from which we obtain
\begin{align}
\nonumber
nR_{k} & \leq I\left( W_{k}; \bar{Y}_{k} \mid W_{\langle k+1:K \rangle},\mathcal{G} \right) + o(n) \\
\label{eq:fano}
& = H\left( \bar{Y}_{k} \mid W_{\langle k+1:K \rangle},\mathcal{G} \right) -
H\left( \bar{Y}_{k} \mid W_{\langle k:K \rangle},\mathcal{G} \right) + o(n)
\end{align}
where $\bar{Y}_{k}\triangleq \big(\bar{Y}_{k}^{[1]},\ldots,\bar{Y}_{k}^{[M]} \big)$ and
$W_{\langle i:j \rangle} = \big(W_{i},\ldots,W_{j}\big)$. By ignoring the $o(n)$ term and adding the rate bounds in \eqref{eq:fano} for all $k \in \mathcal{K}$, we obatin
\begin{align}
\nonumber
n\sum_{k=1}^{K}R_{k}  & \leq
H\left( \bar{Y}_{K} \mid \mathcal{G}\right) -
H\left( \bar{Y}_{1} \mid W_{\langle 1:K \rangle},\mathcal{G}\right)
+ \sum_{k=2}^{K} H\left( \bar{Y}_{k-1} \mid W_{\langle k:K \rangle},\mathcal{G}\right) -
H\left( \bar{Y}_{k} \mid W_{\langle k:K \rangle},\mathcal{G}\right) \\
\label{eq:n_sum_R_k_upperbound}
& \leq n M \log \left( \bar{P} \right) +n o\left( \log \left( \bar{P} \right) \right)
+ \sum_{k=2}^{K}
\underbrace{H\left( \bar{Y}_{k-1} \mid W_{\langle k:K \rangle},\mathcal{G}\right) -
H\left( \bar{Y}_{k} \mid W_{\langle k:K \rangle},\mathcal{G}\right)}_{\triangleq H_{k}^{\Delta}}.
\end{align}
In a DoF sense, the inequality in \eqref{eq:n_sum_R_k_upperbound} becomes:
\begin{equation}
\label{eq:sum_DoF_bound_H_delta}
\sum_{k = 1}^{K} d_{k} \leq 1 + \limsup_{P \rightarrow \infty} \limsup_{n \rightarrow \infty}
\sum_{k = 2}^{K} \frac{H_{k}^{\Delta}}{nM\log(\bar{P})}
\end{equation}
The problem hence reduces to bounding the difference of entropies $H_{k}^{\Delta}$ in a DoF sense.
\begin{lemma}
\label{lemma:diff_entropies}
$H_{k}^{\Delta}$, for any $k \in \langle 2: K \rangle$, is upper bounded in the DoF sense as
\begin{equation}
\label{eq:delta_k upperbound}
\limsup_{P \rightarrow \infty} \limsup_{n \rightarrow \infty}
 \frac{H_{k}^{\Delta}}{nM\log(\bar{P})} \leq \frac{1}{M}  \sum_{m=1}^{M}\alpha_{k}^{[m]} = \alpha_{k}.
\end{equation}
\end{lemma}
The proof of Lemma \ref{lemma:diff_entropies} is relegated to Appendix \ref{appendix:proof_lemma_diff_entropies}.
Finally, from \eqref{eq:delta_k upperbound} and \eqref{eq:sum_DoF_bound_H_delta}, the sum-DoF bound in \eqref{eq:sum_DoF_bound}
is obtained, hence concluding the proof of outer bound.
\section{Sufficiency of Total Order for Separability}
\label{sec:suff_total_order}
In this section, we prove the \emph{if} part of Theorem \ref{th:total_order} by showing that
\begin{equation}
\label{eq:if_sep}
\mathbf{A} \in \mathcal{A}_{\mathrm{to}} \implies
\mathcal{D}_{\mathrm{sep}} = \mathcal{D}_{\mathrm{out}}.
\end{equation}
As $\mathcal{D}_{\mathrm{sep}} \subseteq \mathcal{D}_{\mathrm{out}}$ holds from Theorem \ref{th:outer_bound},
it is only required to show that $\mathcal{D}_{\mathrm{out}} \subseteq \mathcal{D}_{\mathrm{sep}}$,
i.e. for any DoF tuple $\mathbf{d} \in \mathcal{D}_{\mathrm{out}}$,
there exists $\mathbf{d}^{[m]} \in \mathcal{D}^{[m]}$, for every $m \in \mathcal{M}$,
such that  $\mathbf{d} = \frac{1}{M} \sum_{m=1}^{M} \mathbf{d}^{[m]}$.

To show the above, we start by shedding new light on the single-subchannel DoF region $\mathcal{D}^{[m]}$,
described in \eqref{eq:DoF_region_m}.
In particular, we revisit and simplify the achievability argument used in \cite{Piovano2017}
by showing that it is sufficient to vary only one power control variable, in contrast to the $K$ variables employed in \cite{Piovano2017}, to achieve all points of the single-subchannel DoF region.
This leads to an equivalent representation of $\mathcal{D}^{[m]}$, and any polyhedron from the class in Definition \ref{def:pol_region},
which we then use in the second part of this section to show that the statement in \eqref{eq:if_sep} holds.
\subsection{Equivalent DoF Region Representation}
\label{subsec:one_power_variable}
As we focus on the single-subchannel case in this part, we may assume, without loss of generality, that
$\alpha_{1}^{[m]} \geq \alpha_{2}^{[m]}  \geq \cdots \geq \alpha_{K}^{[m]}$ for the subchannel $m$ of interest.
The DoF region $\mathcal{D}^{[m]}$ for subchannel $m$, which is described in \eqref{eq:DoF_region_m},
hence becomes 
\begin{equation}
\label{eq:DoF_region_single_subchannel_order}
\mathcal{D}^{[m]} =
\Big\{ \mathbf{d}^{[m]} \in \mathbb{R}_{+}^{K} : \mathbf{d}^{[m]} (\mathcal{S}) \leq  1 + \bm{\alpha}^{[m]}\big(\mathcal{S}
\setminus \{\min \mathcal{S}\} \big), \ \mathcal{S} \subseteq \mathcal{K} \Big\}.
\end{equation}
As alluded to in Section \ref{subsec:intro_MISO_BC}, the fact that the right-hand-side of \eqref{eq:DoF_region_single_subchannel_order} is an outer bound for $\mathcal{D}^{[m]}$ follows as a direct consequence of the sum-DoF upper bound in \cite[Th. 1]{Davoodi2016}.
On the other hand, the achievability of $\mathcal{D}^{[m]}$, proved in \cite{Piovano2017},
is based on rate-splitting with the superposition of private (zero-forcing) and common (multicasting) codewords
\cite{Yang2013,Clerckx2016}.
We present our own take on the approach in \cite{Piovano2017}, focusing on parts most essential for deriving our
alternative simplified representation.
\begin{itemize}
\item First, the DoF region achieved through rate-splitting and power assignment, which we denote by $\mathcal{D}_{\mathrm{RS}}^{[m]\star}$,
is characterized as the set of all DoF tuples $\mathbf{d}^{[m]} = (d_{1}^{[m]},\ldots,d_{K}^{[m]}) \in \mathbb{R}_{+}^{K}$ satisfying
\begin{subequations}
\label{eq:DoF_region_RS}
\begin{align}
  &(d_{1}^{[m]},\ldots,d_{K}^{[m]}) = (d_{1}^{[m](\mathrm{p})},\ldots,d_{K}^{[m](\mathrm{p})}) + (d_{1}^{[m](\mathrm{c})},\ldots,d_{K}^{[m](\mathrm{c})}) \\
  & d_{i}^{[m](\mathrm{p})} \geq 0, \ d_{i}^{[m](\mathrm{c})} \geq 0, \ i \in \mathcal{K} \\
  &d_{i}^{[m](\mathrm{p})} \leq \Big( a_{i}^{[m]} -  \big( \max_{j \neq i}  a_{j}^{[m]}  - \alpha_{i}^{[m]}   \big)^{+}
  \Big)^{+}, \ i \in \mathcal{K} \\
  &\sum_{i \in \mathcal{K}} d_{i}^{[m](\mathrm{c})} \leq 1 - \max_{j \in \mathcal{K}}a_{j}^{[m]} \\
  &0 \leq a_{i}^{[m]} \leq  1, \ i \in \mathcal{K}.
\end{align}
\end{subequations}
In the above, $\mathbf{d}^{[m](\mathrm{p})} = (d_{1}^{[m](\mathrm{p})},\ldots,d_{K}^{[m](\mathrm{p})}) \in \mathbb{R}_{+}^{K}$  and
$ \mathbf{d}^{[m](\mathrm{c})} = (d_{1}^{[m](\mathrm{c})},\ldots,d_{K}^{[m](\mathrm{c})}) \in \mathbb{R}_{+}^{K}$
are the private and common DoF tuples, associated with the private (zero-forcing) and common (multicasting) signals, respectively.
On the other hand, $\mathbf{a}^{[m]} \triangleq (a_{1}^{[m]},\ldots,a_{K}^{[m]}) \in [0,1]^{K}$ are the power control variables associated with the $K$ private signals, i.e. the power assigned to the $k$-th private signal scales as $O(P^{a_{k}})$.
For a more detailed exposition of the above scheme and its achievable DoF region in \eqref{eq:DoF_region_RS}, readers are referred to \cite{Clerckx2016,Piovano2017} and references therein.

It is evident that each DoF tuple $\mathbf{d}^{[m]} \in \mathcal{D}_{\mathrm{RS}}^{[m]\star}$ is achieved through a strategy
identified by a pair $\big(\mathbf{a}^{[m]},\mathbf{d}^{[m](\mathrm{c})}\big)$, where $\mathbf{a}^{[m]}$ governs
the private DoF tuple $\mathbf{d}^{[m](\mathrm{p})}$ and the common sum-DoF $\mathbf{d}^{[m](\mathrm{c})}(\mathcal{K}) = \sum_{i \in \mathcal{K}} d_{i}^{[m](\mathrm{c})}$, while the individual entries of $\mathbf{d}^{[m](\mathrm{c})}$ determine the manner in which the common sum-DoF is assigned across the $K$ users.
\item The achievable DoF region in \eqref{eq:DoF_region_RS} is shown to coincide with the optimal DoF region, i.e. $\mathcal{D}_{\mathrm{RS}}^{[m]\star} = \mathcal{D}^{[m]}$.
This is accomplished in \cite{Piovano2017} by an exhaustive characterization of
the faces describing the polyhedral outer bound in \eqref{eq:DoF_region_single_subchannel_order}, with the aid of
induction to cope with an arbitrary $K$,
and then explicitly tuning the pair $\big(\mathbf{a}^{[m]},\mathbf{d}^{[m](\mathrm{c})}\big)$
to achieve each such face.
\end{itemize}
Now suppose that we further restrict the tuple of power allocation variables in \eqref{eq:DoF_region_RS} such that
$\mathbf{a}^{[m]} = (a^{[m]},\ldots,a^{[m]})$, where $a^{[m]} \in [0,1]$. By doing so, we essentially reduce the $K$ power allocation variables $(a_{1}^{[m]},\ldots,a_{K}^{[m]})$, employed in $\mathcal{D}_{\mathrm{RS}}^{[m]\star}$, to a single variable $a^{[m]}$.
The resulting achievable region, denoted by $\mathcal{D}_{\mathrm{RS}}^{[m]}$, is given by all DoF tuples $\mathbf{d}^{[m]} = (d_{1}^{[m]},\ldots,d_{K}^{[m]}) \in \mathbb{R}_{+}^{K}$
that satisfy
\begin{subequations}
\label{eq:DoF_region_RS_a}
\begin{align}
  &(d_{1}^{[m]},\ldots,d_{K}^{[m]}) = (d_{1}^{[m](\mathrm{p})},\ldots,d_{K}^{[m](\mathrm{p})}) + (d_{1}^{[m](\mathrm{c})},\ldots,d_{K}^{[m](\mathrm{c})}) \\
  & d_{i}^{[m](\mathrm{p})} \geq 0, \ d_{i}^{[m](\mathrm{c})} \geq 0, \ i \in \mathcal{K} \\
  \label{eq:DoF_region_RS_a_3}
  &d_{i}^{[m](\mathrm{p})} \leq \min\{a^{[m]},\alpha_{i}^{[m]}\}, \ i \in \mathcal{K} \\
  &\sum_{i \in \mathcal{K}} d_{i}^{[m](\mathrm{c})} \leq 1 - a^{[m]} \\
  &0 \leq a^{[m]} \leq  1.
\end{align}
\end{subequations}
Due to the additional constraints of $a_{i}^{[m]} = a^{[m]}$, for all $i \in \mathcal{K}$, it is evident that $\mathcal{D}_{\mathrm{RS}}^{[m]} \subseteq \mathcal{D}_{\mathrm{RS}}^{[m]\star}$.
Nevertheless, such restriction turns out to be lossless in the DoF sense.
\begin{lemma}
\label{lem:single_subchannel}
For the single-subchannel case, the achievable DoF region $\mathcal{D}_{\mathrm{RS}}^{[m]}$ described in \eqref{eq:DoF_region_RS_a} coincides with the optimal DoF region $\mathcal{D}^{[m]}$ described in \eqref{eq:DoF_region_single_subchannel_order}.
\end{lemma}
Lemma \ref{lem:single_subchannel} is proved by eliminating all auxiliary variables in \eqref{eq:DoF_region_RS_a}
and showing that the resulting polyhedron coincides with the one in \eqref{eq:DoF_region_single_subchannel_order}.
This is accomplished through a series of simplifying reductions, followed by an inductive Fourier-Motzkin
elimination procedure. The details of this proof are relegated to
Appendix \ref{appendix:proof_lemma_single_subchannel}.
\subsection{Proof of \eqref{eq:if_sep}}
\label{subsec:separability}
Equipped with Lemma \ref{lem:single_subchannel}, we proceed to show  that under the total order condition in \eqref{eq:total_CSIT_order}, we have
$\mathcal{D}_{\mathrm{sep}} = \mathcal{D}_{\mathrm{out}}$.
Due to \eqref{eq:total_CSIT_order}, we have $\max_{i \in \mathcal{S}}\{ \alpha_{i}^{[m]}\} = \alpha_{\min \mathcal{S}}^{[m]}$  for all subchannels $m \in \mathcal{M}$, and hence 
$\mathcal{D}^{[m]}$ in \eqref{eq:DoF_region_m} is now described by \eqref{eq:DoF_region_single_subchannel_order}, for all $m \in \mathcal{M}$.
Therefore, we may restate $\mathcal{D}_{\mathrm{sep}} = \mathcal{D}_{\mathrm{out}}$ as:
\begin{equation}
\label{eq:Minkowski_sum_equality}
\frac{1}{M} \left( \mathcal{D}^{[1]} \oplus \cdots \oplus \mathcal{D}^{[M]}  \right) = \Big\{ \mathbf{d} \in \mathbb{R}_{+}^{K} : \mathbf{d}(\mathcal{S}) \leq 1 + \frac{1}{M} \sum_{m \in \mathcal{M}} \bm{\alpha}^{[m]}\big(\mathcal{S} \setminus \{\min \mathcal{S}\} \big), \ \mathcal{S} \subseteq \mathcal{K} \Big\}.
\end{equation}
From the associativity of additions, including Minkowski additions,
it is sufficient to show that the equality in \eqref{eq:Minkowski_sum_equality} holds for $M = 2$.
Hence, in what follows we focus on showing that
\begin{equation}
\label{eq:Minkowski_sum_equality_M_2}
\frac{1}{2} \left( \mathcal{D}^{[1]} \oplus \mathcal{D}^{[2]}\right) = \Big\{ \mathbf{d} \in \mathbb{R}_{+}^{K} : \mathbf{d}(\mathcal{S}) \leq 1 + \frac{1}{2} \bm{\alpha}^{[1]}\big(\mathcal{S} \setminus \{\min \mathcal{S}\} \big) + \frac{1}{2}  \bm{\alpha}^{[2]}\big(\mathcal{S} \setminus \{\min \mathcal{S}\} \big), \ \mathcal{S} \subseteq \mathcal{K} \Big\}.
\end{equation}
As highlighted at the beginning of this section, it suffices to show that $ \mathcal{D}_{\mathrm{out}} \subseteq \frac{1}{2} \left( \mathcal{D}^{[1]} \oplus \mathcal{D}^{[2]}\right)$, i.e. for any $\mathbf{d} \in \mathcal{D}_{\mathrm{out}}$, there exists $\mathbf{d}^{[1]} \in \mathcal{D}^{[1]}$ and $\mathbf{d}^{[2]} \in \mathcal{D}^{[2]}$
such that $\mathbf{d} \leq \frac{1}{2} \big( \mathbf{d}^{[1]} + \mathbf{d}^{[2]} \big)$.

Recalling that both $\mathcal{D}_{\mathrm{out}}$ and $\mathcal{D}^{[m]}$ belong to the same class of polyhedra in Definition \ref{def:pol_region},
it follows from Lemma \ref{lem:single_subchannel} that $\mathcal{D}_{\mathrm{out}}$ also assumes an equivalent representation
as the one in \eqref{eq:DoF_region_RS_a}.
In particular,  $\mathcal{D}_{\mathrm{out}}$ is equivalent to all DoF tuples $\mathbf{d} = (d_{1},\ldots,d_{K}) \in \mathbb{R}_{+}^{K}$
that satisfy
\begin{subequations}
\label{eq:DoF_region_RS_a_out}
\begin{align}
\label{eq:DoF_region_RS_a_out_1}
  &(d_{1},\ldots,d_{K}) = (d_{1}^{(\mathrm{p})},\ldots,d_{K}^{(\mathrm{p})}) + (d_{1}^{(\mathrm{c})},\ldots,d_{K}^{(\mathrm{c})}) \\
  & d_{i}^{(\mathrm{p})} \geq 0, \ d_{i}^{(\mathrm{c})} \geq 0, \ i \in \mathcal{K} \\
  \label{eq:DoF_region_RS_a_out_3}
  &d_{i}^{(\mathrm{p})} \leq \min\{a,\alpha_{i}\}, \ i \in \mathcal{K} \\
  \label{eq:DoF_region_RS_a_out_4}
  &\sum_{i \in \mathcal{K}} d_{i}^{(\mathrm{c})} \leq 1 - a \\
  &0 \leq a \leq  1.
\end{align}
\end{subequations}
Note that the average CSIT state, with entries used in \eqref{eq:DoF_region_RS_a_out}, is given by
\begin{equation}
\label{eq:av_CSIT_M_2}
\bm{\alpha} = (\alpha_{1},\ldots,\alpha_{K}) = \frac{1}{2}\big(\alpha_{1}^{[1]},\ldots,\alpha_{K}^{[1]}\big) + \frac{1}{2}\big(\alpha_{1}^{[2]},\ldots,\alpha_{K}^{[2]}\big).
\end{equation}
It follows from \eqref{eq:DoF_region_RS_a_out} that for any $\mathbf{d} \in \mathcal{D}_{\mathrm{out}}$,
there exists  $\mathbf{d}^{(\mathrm{p})} \in \mathbb{R}_{+}^{K}$, $\mathbf{d}^{(\mathrm{c})}\in \mathbb{R}_{+}^{K}$, $a \in [0,1]$ and
$\bm{\lambda} = (\lambda_{1}, \ldots, \lambda_{K}) \in \mathbb{R}_{+}^{K}$, with $\bm{\lambda}(\mathcal{K}) = 1$,
such that:
\begin{align}
\label{eq:d_proof_sum}
\mathbf{d} & \leq \mathbf{d}^{(\mathrm{p})} + \mathbf{d}^{(\mathrm{c})} \\
\label{eq:d_p_proof_sum}
\mathbf{d}^{(\mathrm{p})} & = \big( \min\{a,\alpha_{1}\}, \min\{a,\alpha_{2}\}, \ldots,\min\{a,\alpha_{K}\} \big) \\
\label{eq:d_c_proof_sum}
\mathbf{d}^{(\mathrm{c})} &= \big( (1-a)\lambda_{1},\ldots, (1-a)\lambda_{K} \big).
\end{align}
Note that the (possible) looseness in \eqref{eq:d_proof_sum}, compared to \eqref{eq:DoF_region_RS_a_out_1}, is introduced
to compensate for the imposed tightness in \eqref{eq:d_p_proof_sum}  and \eqref{eq:d_c_proof_sum}, compared to
\eqref{eq:DoF_region_RS_a_out_3} and \eqref{eq:DoF_region_RS_a_out_4}, respectively.

From the order $1 = \alpha_{0} \geq \alpha_{1} \geq \alpha_{2} \geq \cdots \geq \alpha_{K} \geq \alpha_{K+1} = 0$,
it follows that there exists $i \in \langle 1:K+1 \rangle$ such that $\alpha_{i-1} \geq a \geq \alpha_{i}$.
Therefore,  $\mathbf{d}^{(\mathrm{p})}$ in \eqref{eq:d_p_proof_sum} can be rewritten as
\begin{equation}
\label{eq:d_p_proof_sum_2}
\mathbf{d}^{(\mathrm{p})}  = \big( \underbrace{a, \ldots, a}_{i-1 \ \text{entries}},\alpha_{i}, \ldots,\alpha_{K} \big).
\end{equation}
From \eqref{eq:av_CSIT_M_2}, we know that $\alpha_{i} = \frac{1}{2}\alpha_{i}^{[1]} + \frac{1}{2}\alpha_{i}^{[2]}$  and
$\alpha_{i-1} = \frac{1}{2}\alpha_{i-1}^{[1]} + \frac{1}{2}\alpha_{i-1}^{[2]}$.
Combining this with the fact that $\alpha_{i}^{[1]} \leq \alpha_{i-1}^{[1]}$ and $\alpha_{i}^{[2]} \leq \alpha_{i-1}^{[2]}$,
which holds due to the total order in \eqref{eq:total_CSIT_order}, it follows that the
interval $[\alpha_{i}, \alpha_{i-1}] \subset \mathbb{R}_{+}$ is equal to a Minkowski sum of two intervals given by:
\begin{equation}
[\alpha_{i}, \alpha_{i-1}] = \frac{1}{2} \cdot \big[\alpha_{i}^{[1]}, \alpha_{i-1}^{[1]} \big] \oplus \frac{1}{2} \cdot \big[\alpha_{i}^{[2]}, \alpha_{i-1}^{[2]}
\big].
\end{equation}
Therefore, we may express the variable $a$, which is in $[\alpha_{i}, \alpha_{i-1}]$, as:
\begin{equation}
\label{eq:a_decomp_proof_sum}
a = \frac{1}{2} \big( a^{[1]} + a^{[2]} \big), \ \text{for some} \ a^{[1]} \in \big[\alpha_{i}^{[1]}, \alpha_{i-1}^{[1]} \big] \ \text{and} \    a^{[2]} \in \big[\alpha_{i}^{[2]}, \alpha_{i-1}^{[2]} \big].
\end{equation}
It follows from \eqref{eq:a_decomp_proof_sum} that $\mathbf{d}^{(\mathrm{p})}$ and $\mathbf{d}^{(\mathrm{c})}$,
given in \eqref{eq:d_p_proof_sum} and \eqref{eq:d_c_proof_sum} respectively, can be decomposed as:
\begin{align}
\label{eq:d_p_decomp_proof_sum}
\mathbf{d}^{(\mathrm{p})} &= \frac{1}{2} \mathbf{d}^{[1](\mathrm{p})} + \frac{1}{2} \mathbf{d}^{[2](\mathrm{p})}, \text{ where }
\mathbf{d}^{[m](\mathrm{p})}  = \big(a^{[m]}, \ldots, a^{[m]},\alpha_{i}^{[m]}, \ldots,\alpha_{K}^{[m]} \big), \ m \in \{1,2\} \\
\label{eq:d_c_decomp_proof_sum}
\mathbf{d}^{(\mathrm{c})} &= \frac{1}{2} \mathbf{d}^{[1](\mathrm{c})} + \frac{1}{2} \mathbf{d}^{[2](\mathrm{c})}, \text{ where }
\mathbf{d}^{[m](\mathrm{c})}  = \big( (1-a^{[m]})\lambda_{1},\ldots, (1-a^{[m]})\lambda_{K} \big), \ m \in
\{1,2\}.
\end{align}
Defining $\mathbf{d}^{[m]} \triangleq \mathbf{d}^{[m](\mathrm{p})}  + \mathbf{d}^{[m](\mathrm{c})}$, $m \in \{1,2\}$, it follows from
\eqref{eq:d_proof_sum}, \eqref{eq:d_p_decomp_proof_sum} and \eqref{eq:d_c_decomp_proof_sum} that
\begin{align}
\mathbf{d} \leq \frac{1}{2} \big( \mathbf{d}^{[1]} + \mathbf{d}^{[2]} \big).
\end{align}

At this point, it only remains to show that $\mathbf{d}^{[1]} \in \mathcal{D}^{[1]}$ and $\mathbf{d}^{[2]} \in \mathcal{D}^{[2]}$.
To this end, we observe that for any $m \in \{1,2\}$, we have $\alpha_{i-1}^{[m]} \geq a^{[m]} \geq \alpha_{i}^{[m]}$  and
$\alpha_{1}^{[m]} \geq  \alpha_{2}^{[m]}  \geq  \cdots \geq \alpha_{K}^{[m]}$, which hold due to
\eqref{eq:a_decomp_proof_sum} and the total order in \eqref{eq:total_CSIT_order}, respectively.
Therefore, $\mathbf{d}^{[m](\mathrm{p})}$  in \eqref{eq:d_p_decomp_proof_sum} is equal to
\begin{equation}
\label{eq:d_p_decomp_proof_sum_2}
\mathbf{d}^{[m](\mathrm{p})} = \big( \min\{a^{[m]},\alpha_{1}^{[m]}\}, \min\{a^{[m]},\alpha_{2}^{[m]}\}, \ldots,\min\{a^{[m]},\alpha_{K}^{[m]}\} \big).
\end{equation}
From $\mathbf{d}^{[m](\mathrm{c})}$ in \eqref{eq:d_c_decomp_proof_sum}, $\mathbf{d}^{[m](\mathrm{p})}$ in \eqref{eq:d_p_decomp_proof_sum_2} and the equivalent representation in \eqref{eq:DoF_region_RS_a}, it follows that the DoF tuple $\mathbf{d}^{[m]} = \mathbf{d}^{[m](\mathrm{p})}  + \mathbf{d}^{[m](\mathrm{c})}$ is in the achievable region $\mathcal{D}_{\mathrm{RS}}^{[m]}$.
Moreover, from Lemma \ref{lem:single_subchannel}, it follows that $\mathbf{d}^{[m]}$ is also in $\mathcal{D}^{[m]}$,
which completes the proof.
\begin{remark}
From the above, we can see that 
$\mathcal{D}_{\mathrm{sep}} = \mathcal{D}_{\mathrm{out}}$ holds under the total order condition in \eqref{eq:total_CSIT_order}.
This is equivalently expressed, after normalizing \eqref{eq:Minkowski_sum_equality} by $1/M$, as
\begin{equation}
\label{eq:Minkowski_sum}
\mathcal{D}^{[1]} \oplus \mathcal{D}^{[2]}  \oplus \cdots \oplus \mathcal{D}^{[M]} =
\Big\{ \mathbf{d} \in \mathbb{R}_{+}^{K} : \mathbf{d}(\mathcal{S}) \leq \sum_{m = 1}^{M} \Big[ 1 + \bm{\alpha}^{[m]}\big(\mathcal{S}
\setminus \{\min \mathcal{S}\} \big) \Big], \ \mathcal{S} \subseteq \mathcal{K} \Big\}.
\end{equation}
It follows that in this case, the linear inequalities that describe the polyhedron given by the Minkowski sum $\mathcal{D}^{[1]} \oplus \mathcal{D}^{[2]} \oplus \cdots \oplus \mathcal{D}^{[M]}$
are simply the direct sums of the corresponding linear inequalities that describe the $M$ constituent polyhedra.
This property is known to hold whenever the constituent polyhedra are polymatroids\footnote{This property was (essentially) used by Tse and Hanly to characterize the capacity region of the fading multiple-access channel (MAC) \cite{Tse1998}. More recently, it was exploited by Sun and Jafar while studying the separability of the parallel IC under treating interference as noise (TIN), from a GDoF perspective \cite{Sun2016}. Sun and Jafar gave examples for regions characterized by sum-GDoF inequalities, yet are non-polymatroidal and do not enjoy this property.} \cite[Th. 3]{McDiarmid1975}.
Interestingly, this direct summability property holds here despite the fact that
$\mathcal{D}^{[1]},\mathcal{D}^{[2]},\ldots,\mathcal{D}^{[M]}$
are not polymatroidal in general (see Appendix \ref{appendix:non_polymatroidality}).
This property translates to a linearity property for the class of polyhedra in Definition \ref{def:pol_region},
given by
\begin{equation}
\label{eq:Minkowski_sum_property}
\mathcal{P}\left(\frac{1}{M}\sum_{m\in \mathcal{M}} \bm{\alpha}^{[m]}\right) = \frac{1}{M} \bigoplus_{m \in \mathcal{M}} \mathcal{P}\left(\bm{\alpha}^{[m]}\right),
\end{equation}
which is shown to hold, through Theorem \ref{th:total_order}, whenever the monotonic order of entries is preserved across all parameter vectors $\bm{\alpha}^{[1]}, \bm{\alpha}^{[2]} , \ldots,\bm{\alpha}^{[M]}$.
\end{remark}
\section{Necessity of Total Order for Separability}
\label{sec:Necessity_of_Total_Order}
In this section, we prove the \emph{only if} part of Theorem \ref{th:total_order}, i.e. we show that
\begin{equation}
\label{eq:only_if_sep}
\mathbf{A} \notin \mathcal{A}_{\mathrm{to}} \implies
\mathcal{D}_{\mathrm{sep}} \subset
\mathcal{D}.
\end{equation}
The above is shown by explicitly characterizing a set of DoF tuples which are achievable, and hence in
$\mathcal{D}$, yet are not in $\mathcal{D}_{\mathrm{sep}}$, and hence cannot be achieved through separate coding over each subchannel, whenever $\mathbf{A} \notin \mathcal{A}_{\mathrm{to}}$.

First, for any pair of distinct users $k,j \in \mathcal{K}$, define
$\mathcal{D}_{\{k,j\}}$ as the set of DoF tuples that satisfy
\begin{subequations}
\label{eq:DoF_region_2_user_proj}
\begin{align}
d_{i} & = 0, \ \forall i \in \mathcal{K} \setminus\{k,j\}\\
0 \leq d_{k} & \leq 1 , \ 0 \leq d_{j} \leq 1  \\
d_{k} + d_{j} & \leq 1 + \min\{\alpha_{k} , \alpha_{j}\}.
\end{align}
\end{subequations}
It is evident that $\mathcal{D}_{\{k,j\}} \subseteq \mathcal{D}_{\mathrm{out}}$.
In fact, $\mathcal{D}_{\{k,j\}}$ is the projection of $\mathcal{D}_{\mathrm{out}}$ on the plane specified by
$\big\{\mathbf{d} : d_{i}  = 0, \ \forall i \in \mathcal{K} \setminus\{k,j\} \big\}$.
Most importantly, this subregion $\mathcal{D}_{\{k,j\}}$ is achievable.
\begin{lemma}
\label{lemma:DoF_region_2_user_proj}
For any pair of distinct users $k,j \in \mathcal{K}$, we have $\mathcal{D}_{\{k,j\}} \subseteq \mathcal{D}$.
\end{lemma}
\begin{proof}
The above lemma is a direct consequence of the achievability argument in \cite{Hao2013}, which in turn, employs ideas from the rate-splitting scheme in \cite{Yang2013} and the $\mathrm{S}_{3}^{3/2}$ scheme in \cite{Tandon2013}.
In particular, by muting all users in $\mathcal{K} \setminus \{k,j\}$, and using the space-frequency rate-splitting transmission scheme
in \cite[Sec. V]{Hao2013} for users $k$ and $j$, the region $\mathcal{D}_{\{k,j\}}$  is achieved.
\end{proof}

Next, we observe that having a CSIT pattern violating the total order condition is equivalent to
the existence of a pair of users and a pair of subchannels such that one user is stronger, in the CSIT sense, than the other
over one subchannel, and weaker over the second subchannel.
\begin{lemma}
\label{lemma:total_order_equiv}
For any partial CSIT pattern $\mathbf{A}$, we have $\mathbf{A} \notin \mathcal{A}_{\mathrm{to}}$  if and only if
\begin{equation}
\label{eq:2_user_no_order}
\exists k,j \in \mathcal{K} \ \text{and} \ l,q \in \mathcal{M} \ \text{s.t.} \
\alpha_{k}^{[l]} > \alpha_{j}^{[l]} \ \text{and} \ \alpha_{k}^{[q]} < \alpha_{j}^{[q]}.
\end{equation}
\end{lemma}
\begin{proof}
The \emph{if} part of the above statement follows directly from the definition of the total order condition.
To verify the \emph{only if} part, we show that whenever \eqref{eq:2_user_no_order} is violated, we must have $\mathbf{A} \in \mathcal{A}_{\mathrm{to}}$.

In particular, suppose that \eqref{eq:2_user_no_order} does not hold
and consider an arbitrary subchannel $l \in \mathcal{M}$.
Moreover, pick an arbitrary pair of distinct users $k,j \in \mathcal{K}$
such that $\alpha_{k}^{[l]}  > \alpha_{j}^{[l]}$.
If no such pair of users exists, then we must have $\alpha_{k}^{[l]} = \alpha_{j}^{[l]}$ for all $k,j \in \mathcal{K}$ over this subchannel $l$,
which in turn does not cause a violation of $\mathbf{A} \in \mathcal{A}_{\mathrm{to}}$.
Otherwise, having $\alpha_{k}^{[l]}  > \alpha_{j}^{[l]}$ over subchannel $l$ implies that $\alpha_{k}^{[m]}  \geq \alpha_{j}^{[m]}$ over all subchannel $m \in \mathcal{M}$, as we have assumed that \eqref{eq:2_user_no_order} does not hold.
Therefore, users $k$ and $j$ are totally ordered.
Since $l$, $k$ and $j$ are arbitrary, then $\mathbf{A} \in \mathcal{A}_{\mathrm{to}}$ must hold.
\end{proof}
Equipped with Lemma \ref{lemma:total_order_equiv}, we proceed by considering an arbitrary
CSIT pattern $\mathbf{A} \notin \mathcal{A}_{\mathrm{to}}$ and choosing users $k,j \in \mathcal{K}$ and subchannels $l,q \in \mathcal{M}$ for which \eqref{eq:2_user_no_order} holds.
We also assume, without loss of generality, that we have the following order of average CSIT qualities
\begin{equation}
\alpha_{k} \geq \alpha_{j}.
\end{equation}
From Lemma \ref{lemma:DoF_region_2_user_proj}, we know that the following set of DoF tuples is achievable
\begin{equation}
\label{eq:DoF_sum_2_user_proj}
\mathcal{D}_{\Sigma\{k,j\}} \triangleq \big\{ \mathbf{d} \in \mathbb{R}_{+}^{K}: d_{k} + d_{j} = 1 + \alpha_{j} \big\} \cap \mathcal{D}_{\{k,j\}}.
\end{equation}
In particular, $\mathcal{D}_{\Sigma\{k,j\}}$ is a nonempty set that consists of DoF tuples that maximize the sum-DoF
of users $k$ and $j$.
On the other hand, if we restrict our attention to the DoF tuples achieved through separate coding over subchannels, i.e.
$\mathbf{d} \in \mathcal{D}_{\mathrm{sep}}$, then the maximum sum-DoF achieved by users $k$ and $j$ is bounded above as follows
\begin{align}
\label{eq:DoF_sum_2_user_sep_1}
d_{k} + d_{j} & \leq 1 + \frac{1}{M} \sum_{m = 1}^{M} \min \big\{\alpha_{k}^{[m]}, \alpha_{j}^{[m]} \big\} \\
\label{eq:DoF_sum_2_user_sep_2}
& \leq   1 + \frac{1}{M} \left[ \min \big\{\alpha_{k}^{[l]}, \alpha_{j}^{[l]} \big\} + \min \big\{\alpha_{k}^{[q]}, \alpha_{j}^{[q]} \big\} \right] + \frac{1}{M} \sum_{m \in \mathcal{M}\setminus\{l,q\} } \alpha_{j}^{[m]}  \\
\label{eq:DoF_sum_2_user_sep_3}
& <  1 + \frac{1}{M} \sum_{m \in \mathcal{M} } \alpha_{j}^{[m]}\\
\label{eq:DoF_sum_2_user_sep_4}
& = 1 + \alpha_{j}
\end{align}
where \eqref{eq:DoF_sum_2_user_sep_3} follows from $\min \big\{\alpha_{k}^{[l]}, \alpha_{j}^{[l]} \big\} = \alpha_{j}^{[l]}$
and $ \min \big\{\alpha_{k}^{[q]}, \alpha_{j}^{[q]} \big\} = \alpha_{k}^{[q]} < \alpha_{j}^{[q]}$ (see \eqref{eq:2_user_no_order}).

From the above, it is evident that for users $k$ and $j$, the sum-DoF achievable through separate coding in \eqref{eq:DoF_sum_2_user_sep_1} is strictly less than the maximum achievable sum-DoF in \eqref{eq:DoF_sum_2_user_sep_4}.
Hence, whenever $\mathbf{A} \notin \mathcal{A}_{\mathrm{to}}$, there exists
a nonempty set of DoF tuples, i.e. $\mathcal{D}_{\Sigma\{k,j\}}$, such that
\begin{equation}
\mathcal{D}_{\Sigma\{k,j\}} \subseteq \mathcal{D} \ \text{and} \
\mathcal{D}_{\Sigma\{k,j\}} \nsubseteq \mathcal{D}_{\mathrm{sep}}.
\end{equation}
Therefore, \eqref{eq:only_if_sep} holds and separate coding cannot be optimal whenever total order is violated.
\section{Conclusion}
In this paper, we studied the DoF region of the multi-subchannel parallel MISO BC under partial CSIT.
To avoid the intractability of this problem in its generality, we take an alternative approach of identifying conditions
under which a simple separate coding strategy, where a single-subchannel-type scheme is employed in each subchannel, is sufficient to achieve the entire DoF region.
We showed that a total order condition on CSIT patterns is both
necessary and sufficient for separability from the entire DoF region perspective.
This condition, which can be made to hold in practical systems with feedback-based CSIT acquisition,
also leads to maximal DoF regions under per-user CSIT budget constraints.
The outer bound used in our proof is derived by extending the AIS approach, proposed by Davoodi and Jafar for the single-subchannel MISO BC under partial CSIT, to accommodate for multiple subchannels.
Moreover, as an auxiliary result used in showing the achievability side of our result,
we provided a new proof for the DoF region of the single-subchannel setting.
This new proof reduces the number of design variables required to achieve the single-subchannel DoF region compared to a previous proof by Piovano and Clerckx.
\appendix
\section{$\mathcal{P}(\bm{\beta})$ is Not Polymatroidal}
\label{appendix:non_polymatroidality} 
For the sake of completeness, we show here that the class of polyhedra in Definition \ref{def:pol_region} is non-polymatroidal in general,
and therefore the direct summability property of polymatroids in \cite[Th. 3]{McDiarmid1975} cannot be used directly in evaluating Minkowski sums of polyhedra in this class, e.g. in showing that the left-hand-side and the right-hand-side of \eqref{eq:Minkowski_sum} coincide.

First, we define the set function $f: 2^{\mathcal{K}} \rightarrow \mathbb{R}_{+}$ as follows
\begin{equation}
\label{eq:set_function}
f(\mathcal{S}) \triangleq
\begin{cases}
0,& \mathcal{S} = \emptyset \\
1 + \bm{\beta}(\mathcal{S}) - \max_{j \in \mathcal{S}} \beta_{j}, &  \mathcal{S} \subseteq \mathcal{K}, \mathcal{S} \neq  \emptyset
\end{cases}.
\end{equation}
By definition, $\mathcal{P}(\bm{\beta})$ is a polymatroid if the set function $f$ satisfies:
\begin{enumerate}
\item $f(\emptyset) = 0$ (normalized)
\item $f(\mathcal{S}) \leq f(\mathcal{T})$ if $\mathcal{S} \subseteq \mathcal{T}$  (nondecreasing)
\item $f(\mathcal{S}) + f(\mathcal{T}) \geq f(\mathcal{S} \cup \mathcal{T}) + f(\mathcal{S} \cap \mathcal{T})$  (submodular).
\end{enumerate}
The first two conditions are clearly satisfied by $f$ in \eqref{eq:set_function}.
Hence, we turn to showing that $f$ is not submodular in general.
Consider two non-empty subsets of $\mathcal{K}$ denoted by $\mathcal{S}$ and $\mathcal{T}$, and assume that $\mathcal{S} \cap \mathcal{T} \neq \emptyset$.
We denote the union $\mathcal{S} \cup \mathcal{T}$ and intersection $\mathcal{S} \cap \mathcal{T}$ by $\mathcal{U}$
and $\mathcal{I}$, respectively.
As $\mathcal{S}$, $\mathcal{T}$, $\mathcal{U}$ and $\mathcal{I}$ are all non-empty, the
submodularity condition in this case becomes
\begin{equation}
\label{eq:submod_0}
\bm{\beta}(\mathcal{U}) - \max_{u \in \mathcal{U}} \beta_{u} + \bm{\beta}(\mathcal{I}) - \max_{i \in \mathcal{I}} \beta_{i}
\leq \bm{\beta}(\mathcal{S}) - \max_{j \in \mathcal{S}} \beta_{j} + \bm{\beta}(\mathcal{T}) - \max_{k \in \mathcal{T}} \beta_{k}.
\end{equation}
However, in general, we have the following set of inequalities
\begin{align}
\label{eq:submod_1}
\bm{\beta}(\mathcal{U}) + \bm{\beta}(\mathcal{I}) - \max_{u \in \mathcal{U}} \beta_{u} - \max_{i \in \mathcal{I}} \beta_{i}
&  =  \bm{\beta}(\mathcal{S}) + \bm{\beta}(\mathcal{T}) - \max_{u \in \mathcal{U}} \beta_{u} - \max_{i \in \mathcal{I}} \beta_{i} \\
\label{eq:submod_2}
& \leq \bm{\beta}(\mathcal{S}) + \bm{\beta}(\mathcal{T}) - \max_{j \in \mathcal{S}} \beta_{j} - \max_{i \in \mathcal{I}} \beta_{i} \\
\label{eq:submod_3}
& \nleq \bm{\beta}(\mathcal{S}) + \bm{\beta}(\mathcal{T}) - \max_{j \in \mathcal{S}} \beta_{j} - \max_{k \in \mathcal{T}} \beta_{k}
\end{align}
where \eqref{eq:submod_1} holds as we have $\bm{\beta}(\mathcal{U}) = \bm{\beta}(\mathcal{S}) + \bm{\beta}(\mathcal{T}) - \bm{\beta}(\mathcal{I})$
in this case, while \eqref{eq:submod_2} holds due to $\max_{j \in \mathcal{S}} \beta_{j} \leq \max_{u \in \mathcal{U}} \beta_{u} $.
For submodularity to hold, the inequality in \eqref{eq:submod_3} must hold as well, which
is not always the case as we have $\max_{k \in \mathcal{T}} \beta_{k}  \geq \max_{i \in \mathcal{I}} \beta_{i} $.

Guided by the above observations, we construct a simple example that violates submodularity.
Take $(\beta_{1},\beta_{2},\beta_{3}) = (1,0.5,0.8)$,
and consider $\mathcal{S} = \{1,2\}$ and $\mathcal{T} = \{2,3\}$.
For this example, the left-hand-side of \eqref{eq:submod_0} is equal to $0.5+0.8$,
while the right-hand-side is equal to $0.5+0.5$.
As \eqref{eq:submod_0} does not hold, $f$ is not submodular in general, and hence $\mathcal{P}(\bm{\beta})$
is not always a polymatroid.
\section{Proof of Lemma \ref{lemma:diff_entropies}}
\label{appendix:proof_lemma_diff_entropies}
Here we present a proof of the bound in \eqref{eq:delta_k upperbound}, which builds upon and extends the AIS approach in \cite{Davoodi2016}.
We assume that $W_{\langle k:K \rangle}$ are fixed as constants in $H_{k}^{\Delta}$, and therefore we suppress them in the following. Note that this has no influence on the outer bound argument.
\subsection{Functional Dependence}
Given the channel realizations associated with user $k-1$, i.e. $\mathcal{G}_{k-1}$,
there are multiple codewords $(\bar{X}_{1},\ldots,\bar{X}_{k})$ that cast the same image in $\bar{Y}_{k-1}$ in general.
Therefore, the mapping from the received signal
$\bar{Y}_{k-1}$ to one of the codewords $(\bar{X}_{1},\ldots,\bar{X}_{k})$ is random.
This mapping is given by $(\bar{X}_{1},\ldots,\bar{X}_{k}) = \mathcal{L}\bigl(\bar{Y}_{k-1}, \mathcal{G}_{k-1}\bigr)$.
Next, we bound the difference of entropies $H_{k}^{\Delta}$ above as
\begin{align}
\nonumber
H_{k}^{\Delta} & \leq H\left( \bar{Y}_{k-1} \mid \mathcal{G}\right) -
H\left( \bar{Y}_{k} \mid \mathcal{G},\mathcal{L}\right) \\
\label{eq:H_diff_upper_bound}
& \leq H\left( \bar{Y}_{k-1} \mid \mathcal{G}\right) -
H\left( \bar{Y}_{k} \mid \mathcal{G},\mathcal{L} = \mathcal{L}_{0} \right)
\end{align}
where $\mathcal{L}_{0}$ is a mapping which minimizes the term
$H\left( \bar{Y}_{k} \mid \mathcal{G},\mathcal{L}\right) $ over the support of $\mathcal{L}$.
In what follows, we fix a deterministic mapping as
\begin{equation}
\label{eq:fun_dependence}
(\bar{X}_{1},\ldots,\bar{X}_{k}) = \mathcal{L}_{0}\bigl(\bar{Y}_{k-1}, \mathcal{G}_{k-1}\bigr)
\end{equation}
which does not influence the term $H(\bar{Y}_{k-1} \mid \mathcal{G})$ and provides an upper bound for $H_{k}^{\Delta}$, as seen in \eqref{eq:H_diff_upper_bound}.
Therefore, we proceed while assuming that $\bar{Y}_{k}$ is a function of $\bar{Y}_{k-1}$ and $\mathcal{G}$, i.e.
$\bar{Y}_{k}\bigl(\bar{Y}_{k-1}, \mathcal{G}\bigr)$.
\subsection{Aligned Image Sets}
For a given channel realization $\mathcal{G}$, the aligned image set is defined as the set of all distinct
signals at user $k-1$ that cast the same image, e.g. $\bar{Y}_{k} \left( \nu, \mathcal{G}  \right)$, at user $k$.
This is expressed as:
\begin{equation}
\label{eq:aligned_image_set}
\mathcal{A}_{\nu} \left(\mathcal{G} \right) \triangleq
\Bigl\{  \bar{y}_{k-1}  \in \bigl\{\bar{Y}_{k-1} \bigr\} : \bar{Y}_{k}  \left( \bar{y}_{k-1}, \mathcal{G} \right) =
\bar{Y}_{k} \left( \nu, \mathcal{G}  \right)  \Bigr\}
\end{equation}
where $\bigl\{\bar{Y}_{k-1} \bigr\}$ denotes the support of $\bar{Y}_{k-1}$.
Following the exact same steps in \cite[Sec. VI.5]{Davoodi2016},
$H_{k}^{\Delta}$ is bounded in terms of the average size of the aligned image sets as
\begin{equation}
\label{eq:H_delta upperbound}
H_{k}^{\Delta} \leq \log \Bigl( \E \bigl[ | \mathcal{A}_{\bar{Y}_{k-1}} \left(\mathcal{G} \right) | \bigr] \Bigr)
\end{equation}
which is made possible due to the functional dependence assumption in \eqref{eq:fun_dependence}.
The problem now becomes to bound the expected cardinality of $\mathcal{A}_{\bar{Y}_{k-1}} \left(\mathcal{G} \right)$, where the
expectation is over $\bar{Y}_{k-1}$ and $\mathcal{G} $.
Note that for a given realization $\bar{Y}_{k-1} = \nu$, we have
\begin{equation}
\label{eq:average_AIS}
\E \bigl[ | \mathcal{A}_{\nu} \left(\mathcal{G} \right) | \bigr] =
\sum_{\lambda \in  \{\bar{Y}_{k-1}\} } \Prob\bigl( \lambda \in \mathcal{A}_{\nu} \left(\mathcal{G} \right) \bigr).
\end{equation}
Next, we bound the probabilities in \eqref{eq:average_AIS}.
\subsection{Probability of Image Alignment}
To facilitate this step, we recall that from the non-degenerate channel model described in Section \ref{subsec:partial_CSIT},
there exists a constant $\Delta $ such that for any
$G_{ki}^{[m]}(t) \in \mathcal{G}$, we have
\begin{equation}
\label{eq:channel_abs_Delta}
0 < \Delta^{-1} \leq |G_{ki}^{[m]}(t)| \leq \Delta  < \infty
\end{equation}
Moreover, the bounded density assumption implies that the peak of the probability density function
of $G_{ki}^{[m]}(t) \in \mathcal{G}_{k}$, conditioned on CSIT, behaves as $f_{\max} \bar{P}^{\alpha_{k}^{[m]}}$.
For notational convenience, we introduce $G_{kk}^{[m]}(t) = 1$, $k \in \mathcal{K}$, $m \in \mathcal{M}$ and $t \in \langle n \rangle$,
to the channel model in \eqref{eq:received signal_det}.

Given $\mathcal{G}_{k-1}$, consider two distinct realizations of $\bar{Y}_{k-1}$, denoted by $\lambda$ and $\nu$, which are produced by the two corresponding realizations of $(\bar{X}_{1},\ldots,\bar{X}_{k})$ denoted by
$(\bar{\lambda}_{1},\ldots,\bar{\lambda}_{k})$ and $(\bar{\nu}_{1},\ldots,\bar{\nu}_{k})$ respectively, such that
$(\bar{\lambda}_{1},\ldots,\bar{\lambda}_{k}) = \mathcal{L}_{0}\bigl(\lambda, \mathcal{G}_{k-1}\bigr)$ and
$(\bar{\nu}_{1},\ldots,\bar{\nu}_{k}) = \mathcal{L}_{0}\bigl(\nu, \mathcal{G}_{k-1}\bigr)$.
We wish to bound the probability of the event that the images of  $(\bar{\lambda}_{1},\ldots,\bar{\lambda}_{k})$ and
$(\bar{\nu}_{1},\ldots,\bar{\nu}_{k}) $ align at user  $k$, i.e. $\lambda \in \mathcal{A}_{\nu} \left(\mathcal{G} \right) $.
For such alignment event, we must have
\begin{equation}
\label{eq:aligned_images_0}
\sum_{i = 1}^{k} \bigl\lfloor G_{ki}^{[m]}(t) \bar{\lambda}_{i}^{[m]}(t) \bigr\rfloor
=
\sum_{i = 1}^{k} \bigl\lfloor G_{ki}^{[m]}(t) \bar{\nu}_{i}^{[m]}(t) \bigr\rfloor, \ t \in \langle n \rangle, m \in \mathcal{M}.
\end{equation}
It can be easily checked that the event in \eqref{eq:aligned_images_0} implies the following event:
\begin{equation}
\label{eq:aligned_images}
\left| \sum_{i = 1}^{k} G_{ki}^{[m]}(t) \left( \bar{\lambda}_{i}^{[m]}(t) - \bar{\nu}_{i}^{[m]}(t) \right)  \right| \leq k
, \ t \in \langle n \rangle, m \in \mathcal{M}.
\end{equation}
For the purpose of bounding above the probability of alignment,
it is sufficient to consider \eqref{eq:aligned_images}.

Consider a given channel use $t$ and subchannel $m$ and in \eqref{eq:aligned_images}.
Moreover, consider a transmit antenna $ i \in \langle k - 1\rangle$ and recall that $G_{kk}^{[m]}(t) = 1$ is fixed.
By fixing the values of $G_{kj}^{[m]}(t)$, $j \in \langle k - 1\rangle \setminus \{i\}$,
the random variables $G_{ki}^{[m]}(t) \big( \bar{\lambda}_{i}^{[m]}(t) - \bar{\nu}_{i}^{[m]}(t) \big)$ must take values in an interval of length no more than $\frac{2k}{\bar{ | \lambda}_{i}^{[m]}(t) - \bar{\nu}_{i}^{[m]}(t) |} \leq
\frac{2K}{\bar{ | \lambda}_{i}^{[m]}(t) - \bar{\nu}_{i}^{[m]}(t) |}$ so that \eqref{eq:aligned_images} holds.
From the bounded density assumption, the probability of such event is bounded above by
$\frac{2Kf_{\max} \bar{P}^{\alpha_{k}^{[m]}}}{\bar{ | \lambda}_{i}^{[m]}(t) - \bar{\nu}_{i}^{[m]}(t) |}$.
Note that this bound holds for any $i \in \langle k - 1 \rangle$.
Hence, by choosing the tightest of such bounds, the probability of alignment for the channel use $t$ and subchannel $m$, denoted by
$\Prob^{[m]}(t)$, is bounded above by
\begin{equation}
\label{eq:prob_bound_t_m}
\Prob^{[m]}(t) \leq  \frac{2Kf_{\max} \bar{P}^{\alpha_{k}^{[m]}}}{\max_{i \in \langle k-1 \rangle }  | \bar{\lambda}_{i}^{[m]}(t) - \bar{\nu}_{i}^{[m]}(t) |}
\end{equation}
Next, we wish to express the bound in  \eqref{eq:prob_bound_t_m} in terms of the realizations of of $\bar{Y}_{k-1}^{[m]}(t)$, i.e. $\lambda^{[m]}(t)$ and $\nu^{[m]}(t)$.
For this purpose, we bound $|\bar{\lambda}^{[m]}(t) - \bar{\nu}^{[m]}(t) |$ above as follows:
\begin{align}
\nonumber
|\bar{\lambda}^{[m]}(t) - \bar{\nu}^{[m]}(t) | & =
\left| \sum_{i = 1}^{k-1} \bigl\lfloor G_{(k-1)i}^{[m]}(t) \bar{\lambda}_{i}^{[m]}(t) \bigr\rfloor -
\sum_{i = 1}^{k} \bigl\lfloor G_{(k-1)i}^{[m]}(t) \bar{\nu}_{i}^{[m]}(t) \bigr\rfloor \right| \\
\nonumber
& \leq (k-1) + \sum_{i = 1}^{k-1}  \left| G_{(k-1)i}^{[m]}(t) \big( \bar{\lambda}_{i}^{[m]}(t)  -  \bar{\nu}_{i}^{[m]}(t) \big) \right|   \\
\label{eq:prob_bound_t_m_2}
& \leq K + K \Delta \max_{i \in \langle k-1 \rangle }  | \bar{\lambda}_{i}^{[m]}(t) - \bar{\nu}_{i}^{[m]}(t) |
\end{align}
where $\Delta$ is the constant in \eqref{eq:channel_abs_Delta}.
By plugging the bound in \eqref{eq:prob_bound_t_m_2} back into \eqref{eq:prob_bound_t_m}, we obtain
\begin{equation}
\label{eq:prob_bound_t_m_3}
\Prob^{[m]}(t)  \leq
\begin{cases}
       \frac{2K^{2} \Delta  f_{\max} \bar{P}^{\alpha_{k}^{[m]}}}{ | \lambda^{[m]}(t) - \nu^{[m]}(t)| - K },
        \ | \lambda^{[m]}(t) - \nu^{[m]}(t)| > K \\
       1, \ \text{otherwise}
\end{cases}
\end{equation}
where we have used the bound $\Prob^{[m]}(t) \leq 1$ to exclude cases where $| \lambda^{[m]}(t) - \nu^{[m]}(t)| \leq K$ in \eqref{eq:prob_bound_t_m_3}.

Now consider the case of all channel uses $t \in \langle n \rangle$ and subchannels $m \in \mathcal{M}$,
the same approach used above for given $t$ and $m$ is employed to bound the probability of alignment as
\begin{align}
\nonumber
\Prob\bigl( \lambda \in \mathcal{A}_{\nu} \left(\mathcal{G} \right) \bigr)
\leq \ &
\prod_{m = 1}^{M} \prod_{t:| \lambda^{[m]}(t) - \nu^{[m]}(t)| > K} \frac{2K^{2} \Delta  f_{\max} \bar{P}^{\alpha_{k}^{[m]}}}{| \lambda^{[m]}(t) - \nu^{[m]}(t)| - K}  \times
\prod_{m = 1}^{M} \prod_{t: | \lambda^{[m]}(t) - \nu^{[m]}(t)| \leq K} 1
\\
\nonumber
\leq \ &  \Bigl( \max\bigl\{ 2K^{2}  \Delta  f_{\max}, 1 \bigr\}  \Bigr)^{nM} \bar{P}^{n\sum_{m=1}^{M}\alpha_{k}^{[m]}} \times \\
\label{eq:probability_image_alignment}
& \prod_{m = 1}^{M} \left[ \prod_{t: | \lambda^{[m]}(t) - \nu^{[m]}(t)| > K} \frac{1}{ | \lambda^{[m]}(t) - \nu^{[m]}(t)| - K }
   \times \prod_{t: | \lambda^{[m]}(t) - \nu^{[m]}(t)| \leq K} 1 \right].
\end{align}
Note that in \eqref{eq:probability_image_alignment}, we have assumed, without loss of generality, that
$\bar{P} \geq 1$.
Moreover, we have left products which are equal to $1$ explicit to facilitate the following step.
\subsection{Bounding the Average Size of Aligned Image Sets and Combining Bounds}
Equipped with the bound on the probability of alignment in \eqref{eq:probability_image_alignment},
we now proceed to bound $\E \bigl[ | \mathcal{A}_{\nu} \left(\mathcal{G} \right) | \bigr]$.
From \eqref{eq:average_AIS} and \eqref{eq:probability_image_alignment}, we obtain
\begin{align}
\nonumber
\E \bigl[ | \mathcal{A}_{\nu} \left(\mathcal{G} \right) | \bigr]
& =
\sum_{\lambda \in  \{\bar{Y}_{k-1}\} } \Prob\bigl( \lambda \in \mathcal{A}_{\nu} \left(\mathcal{G} \right) \bigr) \\
\nonumber
& \leq   \Bigl( \max\bigl\{ 2K^{2} \Delta  f_{\max},1 \bigr\}  \Bigr)^{nM}  \bar{P}^{n\sum_{m=1}^{M}\alpha_{k}^{[m]}} \times \\
\label{eq:average_AIS_bound_0}
 \prod_{m=1}^{M} \prod_{t=1}^{n} & \left[
 \sum_{\lambda^{[m]}(t): | \lambda^{[m]}(t) - \nu^{[m]}(t)|  \leq  K} 1 +
\sum_{\lambda^{[m]}(t): K < | \lambda^{[m]}(t) - \nu^{[m]}(t)|  \leq  Q_{y}}  \frac{1}{ | \lambda^{[m]}(t) - \nu^{[m]}(t)|  - K} \right] \\
\label{eq:average_AIS_bound}
& \leq  \Bigl( \max\bigl\{ 2K^{2} \Delta  f_{\max},1 \bigr\}  \Bigr)^{nM} \bar{P}^{n\sum_{m=1}^{M}\alpha_{k}^{[m]}}
\times \Bigl( \log(\bar{P}) + o\bigl( \log(\bar{P}) \bigr) \Bigr)^{nM}
\end{align}
where $Q_{y} \triangleq K + K \Delta \lceil \bar{P} \rceil $, which is an upper bound on the values taken by
$| \lambda^{[m]}(t) - \nu^{[m]}(t)| $.
The expression in \eqref{eq:average_AIS_bound_0} is obtained by an interchange of sums and products (see \cite[Footnote 3]{Davoodi2017a}),
while \eqref{eq:average_AIS_bound} is obtained using the partial sum of harmonic series, i.e. $\sum_{i = 1}^{a} \frac{1}{i}\leq 1 + \log(a)$.
\subsection{Combining Bounds}
The bound for $\E \bigl[ | \mathcal{A}_{\nu} \left(\mathcal{G} \right) | \bigr] $ in \eqref{eq:average_AIS_bound} holds for all
$\nu \in \bigl\{\bar{Y}_{k-1} \bigr\}$.
Combining this with \eqref{eq:H_delta upperbound}, we obtain the desired bound for the difference of entropies, given by
\begin{equation}
\limsup_{P \rightarrow \infty} \limsup_{n \rightarrow \infty}
 \frac{H_{k}^{\Delta}}{nM\log(\bar{P})} \leq \frac{1}{M}  \sum_{m=1}^{M}\alpha_{k}^{[m]}.
\end{equation}
\section{Proof of Lemma \ref{lem:single_subchannel}}
\label{appendix:proof_lemma_single_subchannel}
Here we prove that the achievable DoF region $\mathcal{D}_{\mathrm{RS}}^{[m]}$, described in \eqref{eq:DoF_region_RS_a},
is equivalent to the DoF region $\mathcal{D}^{[m]}$, given in \eqref{eq:DoF_region_single_subchannel_order}.
As we deal with only a single subchannel throughout this appendix, we drop the superscript $[m]$ for brevity.

In the first step of our proof, we observe  that
the private DoF variables in \eqref{eq:DoF_region_RS_a} can be easily eliminated by
replacing each variable $d_{i}^{(\mathrm{p})}$ with $d_{i} - d_{i}^{(\mathrm{c})}$, for all $i \in \mathcal{K}$.
After this elimination, the set of inequalities in \eqref{eq:DoF_region_RS_a} are equivalently expressed as
\begin{subequations}
  \label{eq:DoF_RS_region_a}
\begin{align}
  &  d_{i}^{(\mathrm{c})} - d_{i}  \leq 0, \ i \in \mathcal{K} \\
  & - d_{i}^{(\mathrm{c})}  \leq 0, \ i \in \mathcal{K} \\
  \label{eq:DoF_RS_region_a_p1}
  &d_{i} - d_{i}^{(\mathrm{c})}  \leq \alpha_{i}, \ i \in \mathcal{K} \\
  \label{eq:DoF_RS_region_a_p2}
  &d_{i} - d_{i}^{(\mathrm{c})}  \leq a, \ i \in \mathcal{K} \\
  \label{eq:DoF_RS_region_a_1}
  &\sum_{i \in \mathcal{K}} d_{i}^{(\mathrm{c})} \leq 1 - a \\
  &0 \leq a \leq  1.
\end{align}
\end{subequations}
In the above, \eqref{eq:DoF_RS_region_a_p1} and \eqref{eq:DoF_RS_region_a_p2} are equivalent to \eqref{eq:DoF_region_RS_a_3}.
Next, we replace the inequality in \eqref{eq:DoF_RS_region_a_1} with the equality
$\sum_{i \in \mathcal{K}} d_{i}^{(\mathrm{c})} = 1 - a $, which in principle yields an achievable
DoF region contained in $\mathcal{D}_{\mathrm{RS}}$, yet turns out to have no influence on our proof.
As a result, we may now eliminate the power allocation variable $a$ in \eqref{eq:DoF_RS_region_a} by replacing it with
$1 - \sum_{i \in \mathcal{K}} d_{i}^{(\mathrm{c})} = 1 - \mathbf{d}^{(\mathrm{c})}(\mathcal{K})$.
This leaves us with the following set of inequalities
\begin{subequations}
\label{eq:D_m_0}
\begin{align}
  d_{i} - d_{i}^{(\mathrm{c})}  & \leq \alpha_{i}, \ i \in \mathcal{K} \\
   - d_{i}^{(\mathrm{c})}       & \leq 0, \ i \in \mathcal{K} \\
   d_{i}^{(\mathrm{c})} - d_{i}  & \leq 0, \ i \in \mathcal{K} \\
  d_{i} + \mathbf{d}^{(\mathrm{c})}\big( \mathcal{K}\setminus \{i\} \big)  & \leq 1, \ i \in \mathcal{K} \\
  \mathbf{d}^{(\mathrm{c})}( \mathcal{K})  & \leq 1.
\end{align}
\end{subequations}
In what follows, we focus on the set of inequalities in \eqref{eq:D_m_0} and eliminate the common DoF
variables $\mathbf{d}^{(\mathrm{c})}$ using Fourier-Motzkin (FM) elimination \cite[Appendix D]{ElGamal2011}.

The proposed FM elimination procedure comprises $K$ steps, where in each step $k \in \mathcal{K}$, we eliminate the common DoF variable
$d^{(\mathrm{c})}_{k}$.
We further complement the elimination procedure with mathematical induction so that it applies to any arbitrary number of users $K$.
To gain insight into the induction hypothesis, we start by manually carrying out the first two steps of the elimination.
\subsection{FM Elimination: Step $1$}
To eliminate $d_{1}^{(\mathrm{c})}$, we first group the set inequalities in \eqref{eq:D_m_0} into the three following categories depending on the presence and sign of $d_{1}^{(\mathrm{c})}$ on the left-hand-side of the inequalities.
\begin{itemize}
\item Inequalities without $d_{1}^{(\mathrm{c})}$:
\begin{subequations}
\label{eq:D_m_1_no}
\begin{align}
  d_{i} - d_{i}^{(\mathrm{c})} &\leq  \alpha_{i}, \ i \in \langle 2:K \rangle \\
   - d_{i}^{(\mathrm{c})}      & \leq  0, \ i \in \langle 2:K \rangle \\
    d_{i}^{(\mathrm{c})} - d_{i}  &  \leq 0, \ i \in \langle 2: K \rangle \\
\label{eq:D_m_1_no_3}
  d_{1} + \mathbf{d}^{(\mathrm{c})}\big( \mathcal{K}\setminus \{1\} \big) & \leq 1.
\end{align}
\end{subequations}
\item Inequalities with $-d_{1}^{(\mathrm{c})}$:
\begin{subequations}
\label{eq:D_m_1_N}
\begin{align}
  d_{1} - d_{1}^{(\mathrm{c})} &\leq  \alpha_{1} \\
   - d_{1}^{(\mathrm{c})}      & \leq  0.
\end{align}
\end{subequations}
\item Inequalities with $+d_{1}^{(\mathrm{c})}$:
\begin{subequations}
\label{eq:D_m_1_P}
\begin{align}
   d_{1}^{(\mathrm{c})} - d_{1}  & \leq 0 \\
  d_{i} + d_{1}^{(\mathrm{c})} + \mathbf{d}^{(\mathrm{c})}\big( \mathcal{K}\setminus \{1,i\} \big) &\leq  1, \ i \in \langle 2:K \rangle \\
   d_{1}^{(\mathrm{c})} + \mathbf{d}^{(\mathrm{c})}\big( \mathcal{K}\setminus \{1\} \big)     & \leq  1.
\end{align}
\end{subequations}
\end{itemize}
Next, we eliminate the variable $d_{1}^{(\mathrm{c})}$ by adding each inequality in \eqref{eq:D_m_1_N} to every inequality in \eqref{eq:D_m_1_P}
(see, e.g., \cite[Appendix D]{ElGamal2011}).
This procedure yields the following set of inequalities:
\begin{subequations}
\label{eq:D_m_1_N_P}
\begin{align}
  \label{eq:D_m_1_N_P_0}
     0 &  \leq \alpha_{1} \\
     - d_{1}  &  \leq 0 \\
d_{1} + d_{i} + \mathbf{d}^{(\mathrm{c})}\big( \mathcal{K}\setminus \{1,i\} \big) &\leq  1 + \alpha_{1}, \ i \in \langle 2:K \rangle \\
d_{i} + \mathbf{d}^{(\mathrm{c})}\big( \mathcal{K}\setminus \{1,i\} \big) &\leq  1 , \ i \in \langle 2:K \rangle \\
\label{eq:D_m_1_N_P_3}
d_{1} + \mathbf{d}^{(\mathrm{c})}\big( \mathcal{K}\setminus \{1\} \big)     & \leq  1 + \alpha_{1} \\
\label{eq:D_m_1_N_P_4}
\mathbf{d}^{(\mathrm{c})}\big( \mathcal{K}\setminus \{1\} \big)     & \leq  1.
\end{align}
\end{subequations}
At this point, we are left with the inequalities in  \eqref{eq:D_m_1_no} and \eqref{eq:D_m_1_N_P}, where
$d_{1}^{(\mathrm{c})}$  has been eliminated.
Inequalities of the type in \eqref{eq:D_m_1_N_P_0} are clearly not useful, and hence we omit them in the following steps.
We observe that \eqref{eq:D_m_1_N_P_3} is redundant as it is implied by \eqref{eq:D_m_1_no_3}.
Moreover, since $d_{1} \geq 0$, the inequality in \eqref{eq:D_m_1_N_P_4} is redundant as it
is also implied by \eqref{eq:D_m_1_no_3}. It follows that at the end of step $1$ (and at the beginning of step $2$), we have
the following set of inequalities
\begin{subequations}
\label{eq:D_m_2}
\begin{align}
- d_{1}  & \leq 0 \\
  d_{i} - d_{i}^{(\mathrm{c})} &\leq  \alpha_{i}, \ i \in \langle 2:K \rangle \\
   - d_{i}^{(\mathrm{c})}      & \leq  0, \ i \in \langle 2:K \rangle \\
 d_{i}^{(\mathrm{c})} - d_{i}  &  \leq 0, \ i \in \langle 2: K \rangle \\
d_{1} + d_{i} + \mathbf{d}^{(\mathrm{c})}\big( \mathcal{K}\setminus \{1,i\} \big) &\leq  1 + \alpha_{1}, \ i \in \langle 2:K \rangle \\
d_{i} + \mathbf{d}^{(\mathrm{c})}\big( \mathcal{K}\setminus \{1,i\} \big) &\leq  1 , \ i \in \langle 2:K \rangle \\
d_{1} + \mathbf{d}^{(\mathrm{c})}\big( \mathcal{K}\setminus \{1\} \big) & \leq 1.
\end{align}
\end{subequations}
\subsection{{FM Elimination: Step $2$}}
For the purpose of eliminating the variable $d_{2}^{(\mathrm{c})}$, we categorize the inequalities in \eqref{eq:D_m_2} as follows:
\begin{itemize}
\item Inequalities without $d_{2}^{(\mathrm{c})}$:
\begin{subequations}
\label{eq:D_m_2_no}
\begin{align}
 - d_{1}  &  \leq 0 \\
  d_{i} - d_{i}^{(\mathrm{c})} &\leq  \alpha_{i}, \ i \in \langle 3:K \rangle \\
   - d_{i}^{(\mathrm{c})}      & \leq  0, \ i \in \langle 3:K \rangle \\
 d_{i}^{(\mathrm{c})} - d_{i}  &  \leq 0, \ i \in \langle 3: K \rangle \\
\label{eq:D_m_2_no_3}
d_{1} + d_{2} + \mathbf{d}^{(\mathrm{c})}\big( \mathcal{K}\setminus \{1,2\} \big) & \leq 1 + \alpha_{1} \\
d_{2} + \mathbf{d}^{(\mathrm{c})}\big( \mathcal{K}\setminus \{1,2\} \big) & \leq 1.
\end{align}
\end{subequations}
\item Inequalities with $-d_{2}^{(\mathrm{c})}$:
\begin{subequations}
\label{eq:D_m_2_N}
\begin{align}
  d_{2} - d_{2}^{(\mathrm{c})} &\leq  \alpha_{2} \\
   - d_{2}^{(\mathrm{c})}      & \leq  0.
\end{align}
\end{subequations}
\item Inequalities with $+d_{2}^{(\mathrm{c})}$:
\begin{subequations}
\label{eq:D_m_2_P}
\begin{align}
d_{2}^{(\mathrm{c})} - d_{2}  &  \leq 0\\
  d_{1} + d_{i} + d_{2}^{(\mathrm{c})} + \mathbf{d}^{(\mathrm{c})}\big( \mathcal{K}\setminus \{1,2,i\} \big) &\leq  1 + \alpha_{1}
  , \ i \in \langle 3:K \rangle \\
  d_{i} + d_{2}^{(\mathrm{c})} + \mathbf{d}^{(\mathrm{c})}\big( \mathcal{K}\setminus \{1,2,i\} \big) &\leq  1, \ i \in \langle 3:K \rangle \\
  d_{1} + d_{2}^{(\mathrm{c})} + \mathbf{d}^{(\mathrm{c})}\big( \mathcal{K}\setminus \{1,2\} \big)     & \leq  1.
\end{align}
\end{subequations}
\end{itemize}
Now we eliminate $d_{2}^{(\mathrm{c})}$ by adding the inequalities in \eqref{eq:D_m_2_N} and \eqref{eq:D_m_2_P}.
This yields:
\begin{subequations}
\label{eq:D_m_2_N_P}
\begin{align}
 - d_{2}  &  \leq 0\\
  d_{1} + d_{2} + d_{i}  + \mathbf{d}^{(\mathrm{c})}\big( \mathcal{K}\setminus \{1,2,i\} \big) &\leq  1 + \alpha_{1} + \alpha_{2}
  , \ i \in \langle 3:K \rangle \\
    d_{1} + d_{i} + \mathbf{d}^{(\mathrm{c})}\big( \mathcal{K}\setminus \{1,2,i\} \big) &\leq  1 + \alpha_{1}
  , \ i \in \langle 3:K \rangle \\
  d_{2} + d_{i}  + \mathbf{d}^{(\mathrm{c})}\big( \mathcal{K}\setminus \{1,2,i\} \big) &\leq  1 + \alpha_{2}, \ i \in \langle 3:K \rangle \\
  \label{eq:D_m_2_N_P_4}
  d_{i}  + \mathbf{d}^{(\mathrm{c})}\big( \mathcal{K}\setminus \{1,2,i\} \big) &\leq  1, \ i \in \langle 3:K \rangle \\
  \label{eq:D_m_2_N_P_5}
  d_{1} + d_{2} + \mathbf{d}^{(\mathrm{c})}\big( \mathcal{K}\setminus \{1,2\} \big)     & \leq  1 + \alpha_{2} \\
  d_{1}  + \mathbf{d}^{(\mathrm{c})}\big( \mathcal{K}\setminus \{1,2\} \big)     & \leq  1.
\end{align}
\end{subequations}
After eliminating $d_{2}^{(\mathrm{c})}$, we are left with the inequalities in \eqref{eq:D_m_2_no} and
\eqref{eq:D_m_2_N_P}.
Moreover, it can be seen that the inequality in \eqref{eq:D_m_2_no_3} is now redundant as it is implied by the inequality in
\eqref{eq:D_m_2_N_P_5}.
The remaining relevant inequalities in \eqref{eq:D_m_2_no} and
\eqref{eq:D_m_2_N_P} are expressed in compact form as follows:
\begin{subequations}
\label{eq:D_m_3}
\begin{align}
  - d_{i}  & \leq 0, \ i \in \{1,2\}\\
  d_{i} - d_{i}^{(\mathrm{c})} &\leq  \alpha_{i}, \ i \in \langle 3:K \rangle \\
   - d_{i}^{(\mathrm{c})}      & \leq  0, \ i \in \langle 3:K \rangle \\
 d_{i}^{(\mathrm{c})} - d_{i}  &  \leq 0, \ i \in \langle 3:K \rangle \\
\label{eq:D_m_3_3}
\mathbf{d}(\mathcal{S}) + d_{i}  + \mathbf{d}^{(\mathrm{c})}\big( \mathcal{K}\setminus \{1,2,i\} \big) &\leq  1 + \bm{\alpha}(\mathcal{S}), \
\mathcal{S} \subseteq \{1,2\}, i \in \langle 3:K \rangle \\
\label{eq:D_m_3_4}
\mathbf{d}(\mathcal{S})   + \mathbf{d}^{(\mathrm{c})}\big( \mathcal{K}\setminus \{1,2\} \big) &\leq  1 + \bm{\alpha}\big(\mathcal{S} \setminus \{ \min \mathcal{S} \} \big), \ \mathcal{S} \subseteq \{1,2\}.
\end{align}
\end{subequations}
In the above, we use the convention that  $\mathcal{S} = \emptyset$ is a subset of $\{1,2\}$
so that the inequalities in \eqref{eq:D_m_2_N_P_4} are included in \eqref{eq:D_m_3_3}.
On the other hand, by setting $\mathcal{S} = \emptyset$ in \eqref{eq:D_m_3_4}, we obtain the inequality
$\mathbf{d}^{(\mathrm{c})}\big( \mathcal{K}\setminus \{1,2\} \big) \leq  1 $, which is implied by \eqref{eq:D_m_2_N_P_5}
and hence has no influence.
\subsection{FM Elimination: Step $k + 1$}
Guided by the first two elimination steps, we now construct the induction hypothesis.
Suppose that after $k$ steps of the FM procedure, where $k \in \langle 1 : K-2 \rangle$, the variables $d_{1}^{(\mathrm{c})},\ldots,d_{k}^{(\mathrm{c})}$
are eliminated and we are left with the following set of inequalities:
\begin{subequations}
\label{eq:D_end_of_m}
\begin{align}
 - d_{i}  &  \leq 0, \ i \in \langle 1: k \rangle \\
  d_{i} - d_{i}^{(\mathrm{c})} &\leq  \alpha_{i}, \ i \in \langle k+1:K \rangle \\
   - d_{i}^{(\mathrm{c})}      & \leq  0, \ i \in \langle k+1:K \rangle \\
 d_{i}^{(\mathrm{c})} - d_{i}  &  \leq 0, \ i \in \langle k+1:K  \rangle \\
\mathbf{d}(\mathcal{S}) + d_{i}  + \mathbf{d}^{(\mathrm{c})}\big( \mathcal{K}\setminus \big\{ \{i\}\cup\langle1:k \rangle \big\} \big) &\leq  1 + \bm{\alpha}(\mathcal{S}), \
\mathcal{S} \subseteq \langle 1:k \rangle, i \in \langle k+1:K \rangle \\
\mathbf{d}(\mathcal{S})   + \mathbf{d}^{(\mathrm{c})}\big( \mathcal{K}\setminus \langle 1:k \rangle \big) &\leq  1 + \bm{\alpha}\big(\mathcal{S} \setminus \{\min
\mathcal{S} \} \big), \ \mathcal{S} \subseteq \langle 1:k \rangle.
\end{align}
\end{subequations}
Note that the above hypothesis is consistent with the results from steps $1$ and $2$.
Next, we show that by the end of step $k+1$, the variable $d_{k+1}^{(\mathrm{c})}$ is eliminated and
we obtain a set of inequalities similar to \eqref{eq:D_end_of_m}, except that $k$ in \eqref{eq:D_end_of_m}  is replaced with $k+1$.
For this purpose, we group the inequalities in \eqref{eq:D_end_of_m} into the following three categories:
\begin{itemize}
\item Inequalities without $d_{k+1}^{(\mathrm{c})}$:
\begin{subequations}
\label{eq:D_m_plus_1_no}
\begin{align}
 - d_{i}  & \leq 0, \ i \in \langle 1: k \rangle \\
  d_{i} - d_{i}^{(\mathrm{c})} &\leq  \alpha_{i}, \ i \in \langle k+2:K \rangle \\
   - d_{i}^{(\mathrm{c})}      & \leq  0, \ i \in \langle k+2:K \rangle \\
  d_{i}^{(\mathrm{c})} - d_{i}  & \leq 0, \ i \in \langle k+2:K  \rangle \\
\label{eq:D_m_plus_1_no_3}
\mathbf{d}(\mathcal{S}) + d_{k+1} + \mathbf{d}^{(\mathrm{c})}\big( \mathcal{K}\setminus \langle 1:k+1 \rangle \big) & \leq 1 + \bm{\alpha}(\mathcal{S}), \ \mathcal{S} \subseteq \langle 1:k \rangle.
\end{align}
\end{subequations}
\item Inequalities with $-d_{k+1}^{(\mathrm{c})}$:
\begin{subequations}
\label{eq:D_m_plus_1_N}
\begin{align}
  d_{k+1} - d_{k+1}^{(\mathrm{c})} &\leq  \alpha_{k+1} \\
   - d_{k+1}^{(\mathrm{c})}      & \leq  0.
\end{align}
\end{subequations}
\item Inequalities with $+d_{k+1}^{(\mathrm{c})}$:
\begin{subequations}
\label{eq:D_m_plus_1_P}
\begin{align}
 d_{k+1}^{(\mathrm{c})} - d_{k+1}  &  \leq 0 \\
\mathbf{d}(\mathcal{S}) + d_{i}  + d_{k+1}^{(\mathrm{c})} + \mathbf{d}^{(\mathrm{c})}\big( \mathcal{K} \! \setminus \! \big\{ \{i\}\cup \langle 1:k+1 \rangle \big\} \big) & \leq \!  1 \! + \! \bm{\alpha}(\mathcal{S}),
\mathcal{S} \! \subseteq \! \langle 1:k \rangle, i \! \in \! \langle  k+2 \! : \! K  \rangle \\
\mathbf{d}(\mathcal{S}) + d_{k+1}^{(\mathrm{c})}  + \mathbf{d}^{(\mathrm{c})}\big( \mathcal{K}\setminus \langle 1:k+1 \rangle \big) &\leq  1 + \bm{\alpha}\big(\mathcal{S} \setminus \{\min
\mathcal{S} \} \big), \ \mathcal{S} \subseteq \langle 1:k \rangle.
\end{align}
\end{subequations}
\end{itemize}
Now we eliminate $d_{k+1}^{(\mathrm{c})}$ by adding the inequalities in \eqref{eq:D_m_plus_1_N} and \eqref{eq:D_m_plus_1_P},
from which we obtain
\begin{subequations}
\label{eq:D_m_plus_1_N_P}
\begin{align}
 - d_{k+1}  &  \leq 0 \\
\label{eq:D_m_plus_1_N_P_1}
\mathbf{d}(\mathcal{S}) \! + \! d_{k+1} \! + \! d_{i} \!  + \! \mathbf{d}^{(\mathrm{c})}\big( \mathcal{K} \! \setminus \! \big\{ \{i\}\cup \langle 1:k+1 \rangle \big\} \big) \! & \leq \!  1 \! + \! \bm{\alpha}(\mathcal{S} \! \cup \! \{k+1\}),
\mathcal{S} \! \subseteq \! \langle 1\!:\!k  \rangle, i \! \in \! \langle  k\!+\!2 \! : \! K  \rangle \\
\label{eq:D_m_plus_1_N_P_2}
\mathbf{d}(\mathcal{S}) + d_{i}  + \mathbf{d}^{(\mathrm{c})}\big( \mathcal{K}  \setminus  \big\{ \{i\}\cup \langle 1:k+1 \rangle \big\} \big) & \leq   1  + \bm{\alpha}(\mathcal{S}), \ \mathcal{S}  \subseteq  \langle 1:k \rangle, i \in  \langle  k+2 :  K  \rangle \\
\label{eq:D_m_plus_1_N_P_3}
\mathbf{d}(\mathcal{S}) + d_{k+1}  + \mathbf{d}^{(\mathrm{c})}\big( \mathcal{K}\setminus \langle 1:k+1 \rangle  \big) &\leq  1 + \bm{\alpha}\big(\mathcal{S} \cup  \{k+1\} \setminus \{\min
\mathcal{S} \} \big),  \mathcal{S} \subseteq \langle 1:k \rangle \\
\label{eq:D_m_plus_1_N_P_4}
\mathbf{d}(\mathcal{S})  + \mathbf{d}^{(\mathrm{c})}\big( \mathcal{K}\setminus \langle 1:k+1 \rangle  \big) &\leq  1 + \bm{\alpha}\big(\mathcal{S} \setminus \{\min
\mathcal{S} \} \big), \ \mathcal{S} \subseteq \langle 1:k \rangle.
\end{align}
\end{subequations}
After the elimination, we are left with the inequalities in \eqref{eq:D_m_plus_1_no} and \eqref{eq:D_m_plus_1_N_P}.
Next, we observe that for any $\mathcal{S} \subseteq \langle 1:k \rangle$, we have $k+1 > j$ (and hence $\alpha_{k+1} \leq \alpha_{j}$) for all $j \in \mathcal{S}$.
Therefore, it follows that
\begin{equation}
\label{eq:m_plus_1_redundant}
\bm{\alpha}\big(\mathcal{S} \cup  \{k+1\} \setminus \{\min
\mathcal{S} \} \big) = \bm{\alpha}(\mathcal{S}) + \alpha_{k+1} - \max_{j \in \mathcal{S}} \alpha_{j}  \leq \bm{\alpha}(\mathcal{S}), \ \forall \mathcal{S} \subseteq \langle 1:k \rangle.
\end{equation}
From \eqref{eq:m_plus_1_redundant}, we conclude that the inequalities in \eqref{eq:D_m_plus_1_no_3} are redundant as they are implied by the inequalities in \eqref{eq:D_m_plus_1_N_P_3}.
It follows that at the end of step $k+1$, the variable $d_{k+1}^{(\mathrm{c})}$ is eliminated and
we are left with the set of inequalities given by:
\begin{subequations}
\label{eq:D_end_of_m_plus_1}
\begin{align}
  - d_{i}  &  \leq 0, \ i \in \langle 1: k + 1 \rangle\\
  d_{i} - d_{i}^{(\mathrm{c})} &\leq  \alpha_{i}, \ i \in \langle k+2:K \rangle \\
   - d_{i}^{(\mathrm{c})}      & \leq  0, \ i \in \langle k+2:K \rangle \\
  d_{i}^{(\mathrm{c})} - d_{i}  &  \leq 0, \ i \in \langle k+2:K  \rangle \\
\label{eq:D_end_of_m_plus_1_3}
\mathbf{d}(\mathcal{S}) + d_{i}  + \mathbf{d}^{(\mathrm{c})}\big( \mathcal{K}\setminus \big\{ \{i\}\cup \langle 1:k+1 \rangle \big\} \big) &\leq  1 + \bm{\alpha}(\mathcal{S}), \
\mathcal{S} \subseteq \langle 1:k+1 \rangle, i \in \langle k+2:K \rangle \\
\label{eq:D_end_of_m_plus_1_4}
\mathbf{d}(\mathcal{S})   + \mathbf{d}^{(\mathrm{c})}\big( \mathcal{K}\setminus \langle 1:k+1 \rangle \big) &\leq  1 + \bm{\alpha}\big(\mathcal{S} \setminus \{\min
\mathcal{S} \} \big), \ \mathcal{S} \subseteq \langle 1:k+1 \rangle.
\end{align}
\end{subequations}
Note that \eqref{eq:D_end_of_m_plus_1_3} corresponds to \eqref{eq:D_m_plus_1_N_P_1} and \eqref{eq:D_m_plus_1_N_P_2},
while \eqref{eq:D_end_of_m_plus_1_4} corresponds to \eqref{eq:D_m_plus_1_N_P_3} and \eqref{eq:D_m_plus_1_N_P_4}.
It is evident that the set of inequalities in \eqref{eq:D_end_of_m_plus_1} take the same form of the set of inequalities in \eqref{eq:D_end_of_m},
with the difference that $k+1$ replaces $k$.
\subsection{FM Elimination: Step $K$}
From the above induction hypothesis, by setting $k = K-2$, it can be seen that at the end of step $k+1 = K-1$ of the FM procedure,
we obtain the following set of inequalities:
\begin{subequations}
\label{eq:D_start_of_K}
\begin{align}
  - d_{i}  &  \leq 0, \ i \in \langle 1: K - 1 \rangle \\
  d_{K} - d_{K}^{(\mathrm{c})} &\leq  \alpha_{K} \\
   - d_{K}^{(\mathrm{c})}      & \leq  0 \\
    d_{K}^{(\mathrm{c})} - d_{K}  &  \leq 0 \\
\mathbf{d}(\mathcal{S}) + d_{K}   &\leq  1 + \bm{\alpha}(\mathcal{S}), \
\mathcal{S} \subseteq \langle 1:K-1 \rangle \\
\mathbf{d}(\mathcal{S})   + d^{(\mathrm{c})}_{K} &\leq  1 + \bm{\alpha}\big(\mathcal{S} \setminus \{\min
\mathcal{S} \} \big), \ \mathcal{S} \subseteq \langle 1:K-1 \rangle.
\end{align}
\end{subequations}
Therefore, after eliminating $d_{K}^{(\mathrm{c})}$ in step $K$, we are left with the following set of inequalities:
\begin{subequations}
\label{eq:D_end_of_K}
\begin{align}
\label{eq:D_end_of_K_0}
   - d_{i}  &  \leq 0, \ i \in \mathcal{K}\\
\label{eq:D_end_of_K_1}
\mathbf{d}\big(\mathcal{S}' \cup \{K\} \big)  &\leq  1 + \bm{\alpha}(\mathcal{S}'), \
\mathcal{S}' \subseteq \langle 1:K-1 \rangle \\
\label{eq:D_end_of_K_2}
\mathbf{d}(\mathcal{S})   &\leq  1 + \bm{\alpha}\big(\mathcal{S} \setminus \{\min
\mathcal{S} \} \big), \ \mathcal{S} \subseteq \mathcal{K}.
\end{align}
\end{subequations}
Finally, we show that the set of inequalities in \eqref{eq:D_end_of_K_1} are redundant.
For $\mathcal{S}' = \emptyset$ in \eqref{eq:D_end_of_K_1}, it can be seen that the resulting inequality is included in \eqref{eq:D_end_of_K_2}.
Therefore, we consider a non-empty subset $\mathcal{S}' \subseteq \langle 1:K-1 \rangle$ in \eqref{eq:D_end_of_K_1}
and choose $\mathcal{S} = \mathcal{S}' \cup \{K\}$ in \eqref{eq:D_end_of_K_2} to obtain the corresponding inequality.
Since $K > j$ (and hence $\alpha_{K} \leq \alpha_{j}$) for all $j \in \mathcal{S}'$, we have
\begin{align}
\label{eq:K_redundant}
\bm{\alpha}\big(\mathcal{S}' \cup \{K\} \setminus \{\min \{
\mathcal{S}',K\} \} \big) & = \bm{\alpha}\big(\mathcal{S}' \cup \{K\} \setminus \{\min \mathcal{S}' \} \big)
= \bm{\alpha}(\mathcal{S}') + \alpha_{K} - \max_{j\in \mathcal{S}'}\alpha_{j}
\leq \bm{\alpha}(\mathcal{S}').
\end{align}
Hence, we conclude that the inequalities in \eqref{eq:D_end_of_K_1} are looser in general compared to the corresponding inequalities in \eqref{eq:D_end_of_K_2}.
This leaves us with \eqref{eq:D_end_of_K_2} in addition to the non-negativity conditions in \eqref{eq:D_end_of_K_0}.
Therefore, $\mathcal{D}_{\mathrm{RS}}^{[m]}$ in \eqref{eq:DoF_region_RS_a} is equivalent $\mathcal{D}^{[m]}$ in \eqref{eq:DoF_region_single_subchannel_order}, which concludes the proof.
\section*{Acknowledgement}
The authors would like to thank the anonymous reviewers for their valuable comments.
\bibliographystyle{IEEEtran}
\bibliography{References}
\end{document}